\documentclass[runngingheads]{llncs}
\usepackage[T1]{fontenc}
\usepackage{tikz}
\usepackage{pgfplots}
\usepgfplotslibrary{statistics}
\usetikzlibrary{automata, arrows.meta, positioning}
\usepackage{stmaryrd}
\usepackage{subcaption}
\usepackage{wrapfig}
\usepackage{amsmath,amsfonts,amssymb}

\newcommand{\ini}{\mathfrak{i}}
\newcommand{\sem}[1]{\llbracket #1 \rrbracket}
\newcommand\recht\operatorname

\newtheorem{conclusion}{Conclusion}

\allowdisplaybreaks

\begin{document}

\title{Towards realistic large random models of labeled transition systems and their 0-1 laws} 



\author{Milan {Lopuha\"a-Zwakenberg}}

\authorrunning{Milan LZ}

\institute{University of Twente, the Netherlands\\
\email{m.a.lopuhaa@utwente.nl}}

\maketitle              
\begin{abstract}
Model checking is the automated verification of properties (specified in some modal logic) in labeled transition systems (LTSs); it is an essential tool in ensuring software systems function as intended. State spaces of software grow exponentially, and heuristics are needed to ensure model checking remains feasible in real-world applications. Heuristics, in turn, require a good understanding on the typical behaviour of LTSs.

In this paper, we use random graph theory to create a probabilistic model of large LTSs. From a theoretical analysis of the creation of large LTSs, backed by empirical data from the Model Checking Contest, we endow these models with realistic parameter values.

Then, we analyze the asymptotic behaviour of this model under LTL and CTL, two modal logics popular in model checking. We show that, depending on the precise model, as the size grows to infinity we either have a convergence law (for every formula, the probability that it holds converges to a limit) or a 0-1 law (...and this limit is 0 or 1). We also discuss the theoretical complexity of determining these limits, and give algorithms for doing so. These results are the starting point towards a deep theoretical understanding of typical LTS behaviour, and highlight the promising applicability of random graph theory to model checking.
\end{abstract}

\section{Introduction}

\subsubsection{Model checking.} Model checking is the automated verification of properties (expressed as arbitrary properties in a given logic) of labeled transition systems (LTSs). Typically, these LTSs are models of software systems, and the logics are modal logics expressing properties of the software's possible runs; this makes model checking an important tool in ensuring bug-free software \cite{bauch2014ltl,chaki2004state}.

The fundamental challenge of model checking is the \emph{state-space explosion}: the size of the LTS model of a software system grows exponentially in its number of lines of code. Therefore, computational efficiency is essential for model checkers to meet the demands of industry. An important tool in this respect are heuristics: algorithms that exploit typical properties of large LTSs, making them work well in practice even if the worst-case behaviour is still the same. Such methods include identifying more probable paths in exhaustive searches \cite{edelkamp2006large} or merging equivalent states in state spaces \cite{jensen2022automata}. Developing heuristics requires a good understanding of typical large LTSs, leading to the following question:\\

\noindent \textbf{Question.} \emph{What is the behaviour of realistic large LTSs?}\\

Unfortunately, no good theoretical framework exists for answering this question. In its absence, heuristics are usually evaluated on benchmarks, such as the Model Checking Contest benchmark \cite{amat2023behind}. This can provide solid empirical evidence of the effectiveness of heuristics, but it has the potential to lead to `overfitting': heuristics too tailored to the benchmarks to work well in more general situations. Furthermore, a good understanding of the theory behind realistic large LTS can lead to the development of new heuristics.

\noindent \textbf{Random graphs.} Random graph theory creates and studies probabilistic models of large graphs; it has been used as a model for many large real-world networks and interactions, including wireless networks \cite{lu2013exploring}, epidemiology \cite{britton2020mathematical}, and social media \cite{adriaans2018weighted}. The core philosophy is that for very large networks, the randomness of the probabilistic model cancels out, and the `typical' behaviour of the network emerges. In practice, many real-world networks can be described by relatively simple probabilistic models, which means that studying their properties is mathematically tractable. Since LTSs are directed graphs with some extra decoration, it makes sense to look for random graph models for LTSs.

\noindent \textbf{0-1 laws.} Rather than specific mathematical properties, model checking is concerned with the behaviour w.r.t. arbitrary formulas in a given logic. In random graph theory, this behaviour can often be codified in terms of \emph{0-1 laws}. A 0-1 law for a logic $\mathcal{L}$ and a random graph model $G_n$ of size $n$, states that for every $\varphi \in \mathcal{L}$, as $n \rightarrow \infty$, we either have $\mathbb{P}(G_n \models \varphi) \rightarrow 0$ or $\mathbb{P}(G_n \models \varphi) \rightarrow 1$. Hence practically every instance of $G_n$ behaves the same; the 0-1 law codifies the random model's typical behaviour. 0-1 laws were first proven for Erd\H{o}s-Rényi models \cite{fagin1976probabilities}, the most elementary random graphs, and have since been extended both to more elaborate graph models \cite{shelah1988zero} and to more general mathematical structures and logics \cite{KOLAITIS1992258}. In some cases, rather than a 0-1 law,  a weaker \emph{convergence law} holds, stating that $\mathbb{P}(G_n \models \varphi)$ has a limit as $n \rightarrow \infty$ \cite{lynch2005convergence}.

From the perspective of random graphs and 0-1 laws/convergence laws, the question above can be rephrased as follows:\\

\noindent \textbf{Question.} \emph{What is a good probabilistic model for large LTSs, and do modal logics follow 0-1 laws or convergence laws on this model?}\\

Our aim is to answer this question. Specifically, we consider \emph{Kripke structures} (KS) as our LTS model, and \emph{Linear Temporal Logic} (LTL) and \emph{Computational Tree Logic} (CTL) as our modal logics. These LTSs and logics play an important role in model checking of software and hardware systems \cite{schneider1995verifying,chaki2004state,baier2008principles,classen2010ctl}.

\noindent \textbf{Existing work.} 0-1 laws have been studied for so-called \emph{$\mathcal{S}$-structures}, a very general mathematical object that also includes KSs. 0-1 laws for $\mathcal{S}$-structures have been developed for many logics \cite{Compton,GRANDJEAN1983180,fagin1976probabilities,Kaufmann2001,KAUFMANN1985285}. In particular, random $\mathcal{S}$-structures follow a 0-1 law under iterative fixed-point logic \cite{BLASS198570}, which contains LTL and CTL. The work \cite{dong20250} further expands on this, giving new proofs tailored to the model checking setting and studying the computational complexity of deciding, for a given formula, between the `0' and `1' of a 0-1 law.

However, all these approaches only consider the basic Erd\H{o}s-Rényi model of random graphs: all transitions (and in general $\mathcal{S}$-structures, other features) have the same existence probability $p$, which is contant in $n$, and the events of their existence are all independent. This model is attractive from a theoretical perspective, but it has not been evaluated whether this model fits real-world KSs well. In fact, as we will see in Section \ref{sec:prob}, this model results in KSs that have considerably more transitions and initial states than what one finds in practice. Therefore, a more realistic probabilistic model is needed.

The failure of constant Erd\H{o}s-Rényi graphs to model real-world networks is well known in random graph theory. To account for this, many new models have been created, in order to better model e.g. sparse graphs \cite{hofstad2017random}, graphs with small distances \cite{barabasi1999emergence} and product graphs \cite{leskovec2010kronecker}. However, such models are typically only developed for graphs, not more general structures; and 0-1 laws are only studied for some of these models \cite{shelah1988zero}, and only for first-order logic.

We can conclude that model theory studies very general structures and logics, but only basic probabilistic models; and random graph theory studies a wide range of probabilistic models, but only basic structures and logics. For realistic models of KSs, a new approach is needed.

\subsubsection{Contributions.} Our contributions are twofold: we develop a realistic probabilistic model $\mathbb{K}$ for large KSs, and we prove 0-1 laws for $\mathbb{K}$ for LTL and CTL. For $\mathbb{K}$, we need to define how our probabilistic model affects the presence of state transitions (edges), initial states, and atomic propositions (vertex properties). We take the Erd\H{o}s-Rényi approach in considering all of these independent events and keeping probabilities homogenous; thus each transition is present with probability $p_{\recht{t}}$, each state is initial with probability $p_{\recht{i}}$, and each state satisfies atomic proposition $a$ with probability $p_a$. While this is an oversimplification to some extent, independence is a common feature of random graph models \cite{bollobas2007phase,leskovec2010kronecker,hofstad2024random}. The homogeneity assumption is less essential, but primarily makes computation more convenient; we will return to the impact of this assumption in the conclusion.

The main question is then, how $p_{\recht{t}},p_{\recht{i}},p_a$ behave as functions of the size $n$ of $\mathbb{K}$. To answer this question, we consider the fact that large KSs are created as the \emph{parallel composition} of many small KSs. The parallel composition is more or less equal to the product graph, except that some transitions are shared, allowing the basic KSs to communicate. We show that, under some mild assumptions on the number of shared transitions, this results in 
\[
p_{\recht{t}} \propto \tfrac{\log n}{n}, \quad \quad p_{\recht{i}} = n^{-r},\quad \quad  p_{a} = \textrm{ constant,}
\]
for some $r \in [0,1]$. This behaviour is considerably different from the KSs in \cite{dong20250}, where all $p$ are considered to be constant. Empirical data from the Model Checking Contest \cite{amat2023behind} show that these probabilities accurately model real-world KSs. We also consider the model $\mathbb{K}^1$ with only a single initial state, as this occurs often in practice.

Our main mathematical result is the existence of 0-1 laws for $\mathbb{K}$, and convergence laws for $\mathbb{K}^1$. 0-1 laws fail for $\mathbb{K}^1$, since the validity of a formula inherently depends on the random properties of the single initial vertex.

\begin{theorem}[Theorems \ref{thm:01ltl},\ref{thm:01ltlone},\ref{thm:ctl} paraphrased] $\mathbb{K}$ satisfies a 0-1 law for LTL, and $\mathbb{K}^1$ a convergence law. For large enough parameter choices the same holds for CTL. Computing the limit probabilities is NP-hard for LTL and linear for CTL (for both $\mathbb{K}$ and $\mathbb{K}^1$), and we give algorithms.
\end{theorem}

This result shows that not only do we have the mathematical statement that 0-1 laws exist, we also make them operational by providing algorithms. This lays the foundation towards a better understanding of the behaviour of large, realistic KSs, paving the way for the development of new heuristics. Furthermore, our probabilistic model can serve as a method to generate large, realistic benchmarks to evaluate model checkers.

Interestingly, our results are completely analogous to those of \cite{dong20250} of a much more simplistic probabilistic model; our proofs, however, are completely different. That work crucially relies on \emph{extension axioms}, which roughly state that every KS $K$ can be found in the random structure if $n$ is large enough. This does not hold for our $\mathbb{K}$, which is too sparse for this property. Instead, considerably more intricate reasoning is needed. While existing results from random graph theory cannot be directly used, because these usually do not consider infinite paths, we use many similar strategies, especially with respect to proving connectivity or the existence of certain substructures \cite{hofstad2017random}. Thus, this work highlights the relevance of random graph theory for model checking.

All proof sketches are expanded into full proofs in the appendix.

\section{Kripke structures} \label{sec:KS}

We first recap the definition of Kripke structures (KSs): LTSs with state labels and nondeterministic transitions \cite{kripke1963semantical}. KSs are the typical models considered in software model checking \cite{baier2008principles}. We follow model checking terminology and talk about \emph{states} and \emph{transitions} rather than \emph{vertices} and \emph{edges}/\emph{arcs}.

\begin{definition}
Let $A$ be a finite set of \emph{atomic propositions}. An \emph{$A$-labeled digraph} is a tuple $(V,E,I,\ell)$ where $(V,E)$ is a directed graph (possibly with self-loops) of \emph{states} and \emph{transitions}, and $\ell$ is a map $\ell\colon V \rightarrow 2^A$ and $I \subseteq V$ is the set of \emph{initial states}. A labeled digraph is a \emph{Kripke strucure} if every state has at least one outgoing transition, and if $I$ is nonempty.
\end{definition}

From the perspective of software systems, the \emph{state space} $V$ represents the possible states the program could be in, and $E$ represents the ways in which the software can change states (modeled non-deterministically). The set of initial states $I$ is nonempty, and often has only one element. The condition that every state has at least one outgoing transition is there to prevent deadlocks.

\begin{figure}[t]
\begin{tikzpicture}
\node (v0) [state,accepting] at (1,0) {};
\node (v1) [state] at (1,-2) {$a_1,a_2$};
\path[-stealth,thick]
(v1) edge[color=blue, bend right] node[right] {$m_3$} (v0)
(v1) edge[bend left] node[left] {$m_1$} (v0)
(v1) edge[color=red,loop below] node {$m_2$} ();

\node (q0) [state, accepting] at (4,-5) {};
\node (q1) [state] at (6,-5) {$a_2$};
\node (q2) [state, accepting] at (8,-5) {$a_3$};
\path[-stealth,thick]
(q0) edge[color=red,bend right] node[below] {$m_2$} (q1)
(q1) edge[color=blue,bend right] node[above] {$m_3$} (q0)
(q1) edge node[above] {$m_4$} (q2)
(q2) edge[color=red,loop above] node{$m_2$} ();

\node (w00) [state,accepting] at (4,0) {};
\node (w01) [state] at (6,0) {$a_2$};
\node (w02) [state,accepting] at (8,0) {$a_3$};
\node (w10) [state] at (4,-2) {$a_1,a_2$};
\node (w11) [state] at (6,-2) {$a_1,a_2$};
\node (w12) [state] at (8,-2) {$a_1,a_2,a_3$};
\path[-stealth,thick]
(w10) edge node[left] {$m_1$} (w00)
(w11) edge node[right] {$m_1$} (w01)
(w12) edge node[left] {$m_1$} (w02)
(w01) edge node[above] {$m_4$} (w02)
(w11) edge node[above] {$m_4$} (w12)
(w10) edge[color=red] node[above] {$m_2$} (w11)
(w12) edge[color=red,loop right] node {$m_2$} ()
(w11) edge[color=blue] node[below left] {$m_3$} (w00);

\draw (1,-4) node {$K$};
\draw (3,-5) node {$K'$};
\draw (6,1) node {$K\mathbin{\parallel}K'$};

\draw[dashed] (3,1) -- (3,-3) -- (9,-3);
\end{tikzpicture}
\caption{Example of a parallel composition $K\mathbin{\parallel}K'$. The label sets are $A = \{a_1,a_2\}$, $A' = \{a_2,a_3\}$, $M = \{m_1,m_2,m_3\}$, $M' = \{m_2,m_3,m_4\}$. The elements of the set $\ell(v)$ are inscribed within each state $v$. Shared actions are coloured.} \label{fig:parco}
\end{figure}
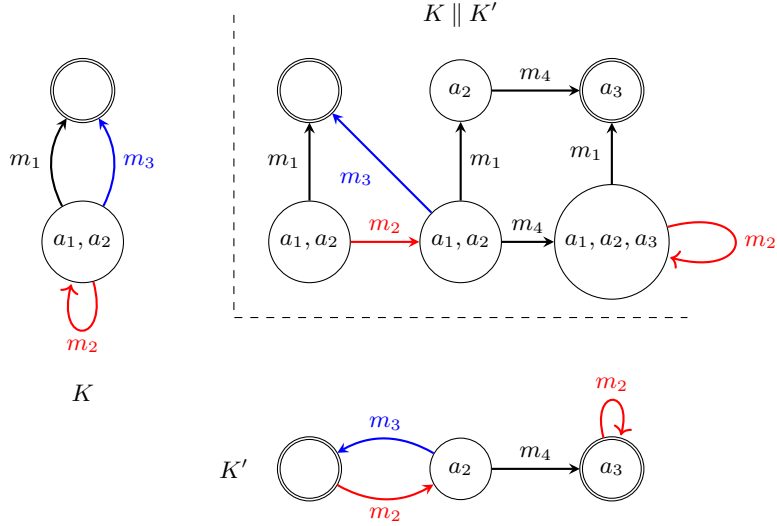

State spaces are typically very large: sizes of $>10^{100}$ are not out of the ordinary. Clearly, such models are impossible to validate. Instead, one usually obtains large KSs by composing smaller KSs that are easier to check by hand. The main tool is \emph{parallel composition}. This is similar to a product graph, except that some transitions are shared between the components to represent communication. To define which transitions are shared, we have to label them:

\begin{definition}
Let $A$ and $M$ be finite sets. An \emph{$(A,M)$-bilabeled digraph} is a tuple $(V,(E_m)_{m \in M},I,\ell)$ where $V$ is a set, each $E_m \subseteq V \times V$ is a subset, $I \subseteq V$ is a subset, and $\ell\colon V \rightarrow 2^A$ is a function. Its associated $A$-labeled digraph is $(V,I,\bigcup_mE_m,\ell)$.
\end{definition}

$E_m$ is the set of transitions that can be taken via action $m$; to go from a bilabeled digraph to a labeled one simply means forgetting the action corresponding to each transition.

To get the parallel composition, we take the product state space, and we add every transition whose label occurs only in one of the two transition label sets. For the shared labels, transitions have to be taken in both components simultaneously. Formally, this is defined as follows. We write $\Delta_V = \{(v,v) \mid v \in V\} \subseteq V \times V$.

\begin{definition} \label{def:parco}
Let $A,A',M,M'$ be finite sets. Let $K = (V,(E_m)_{m\in M},I,\ell)$ be an $(A,M)$-bilabeled digraph, and let $K' = (V',I',(E'_m)_{m \in M'},\ell')$ be an $(A',M')$-bilabeled digraph. Then the \emph{parallel composition} $K \mathbin{\parallel} K'$ is the $(A \cup A',M \cup M')$-bilabeled digraph  $(V'',(E''_m)_{m \in M \cup M},I'',\ell'')$ defined as follows:
\begin{align*}
V'' &= V \times V',\\
E''_m &= \begin{cases}
E_m \times \Delta_{V'}, & \textrm{ if $m \in M \setminus M'$},\\
\Delta_V \times E'_m, & \textrm{ if $m \in M \setminus M'$},\\
E_m \times E'_m, & \textrm{ if $m \in M \cap M'$}.
\end{cases}\\
I'' &= I \times I',\\
\ell''(v,v') &= \ell(v) \cup \ell'(v').
\end{align*}
\end{definition}

The definition of $E''_m$ can be interpreted as follows: if $m$ is a non-shared action, it performs a transition only on the digraph it belongs to. If it is shared, it must operate on both digraphs simultaneously. An example is given in Fig.~\ref{fig:parco}. Note that the parallel composition of two KSs is not necessarily a KS. Furthermore, $\parallel$ is a commutative and associative operation on KSs (up to equivalence).

\section{Realistic parameter scaling} \label{sec:prob}

Before we give our probabilistic model of large KSs in Section \ref{sec:prelim}, we first discuss the behaviour of large KSs in practice. We care about three parameters, as functions of the size $n$: their average degree, their number of initial states, and the number of states satisfying each atomic proposition.

\subsection{Average degree}

\begin{figure}[t]
\begin{tikzpicture}
\node[circle,draw, fill=black!30, minimum size = 0.2cm, inner sep=0pt] (000) at (0,0) {};
\node[circle,draw, fill=black!30, minimum size = 0.2cm, inner sep=0pt] (100) at (3,0) {};
\node[circle,draw, fill=black!30, minimum size = 0.2cm, inner sep=0pt] (010) at (0,3) {};
\node[circle,draw, fill=black!30, minimum size = 0.2cm, inner sep=0pt](110) at (3,3) {};
\node[circle,draw, fill=black!30, minimum size = 0.2cm, inner sep=0pt] (001) at (2,1) {};
\node[circle,draw, fill=black!30, minimum size = 0.2cm, inner sep=0pt] (101) at (5,1) {};
\node[circle,draw, fill=black!30, minimum size = 0.2cm, inner sep=0pt] (011) at (2,4) {};
\node[circle,draw, fill=black!30, minimum size = 0.2cm, inner sep=0pt] (111) at (5,4) {};
\path[-stealth,thick]
(010) edge node[right] {$m_1$} (000)
(011) edge node[right] {$m_1$} (001)
(110) edge node[right] {$m_1$} (100)
(111) edge node[right] {$m_1$} (101)
(000) edge node[above] {$m_4$} (001)
(010) edge node[above] {$m_4$} (011)
(100) edge node[above] {$m_4$} (101)
(110) edge node[above] {$m_4$} (111)
(000) edge[color=red] node[below] {$m_2$} (100)
(011) edge[color = blue, loop above] node[above] {$m_3$} ()
(001) edge[color = blue, loop below] node[below] {$m_3$} ()
(111) edge[color = blue] node[above] {$m_3$} (011)
(101) edge[color = blue] node[above] {$m_3$} (001);
\node[circle,draw, fill=black!30, minimum size = 0.2cm, inner sep=0pt] (x0) at (2,6) {};
\node[circle,draw, fill=black!30, minimum size = 0.2cm, inner sep=0pt] (x1) at (5,6) {};
\path[-stealth,thick]
(x0) edge[color = blue, loop left] node {$m_3$} ()
(x0) edge[color = red, bend left] node[above] {$m_2$} (x1)
(x1) edge[color = blue] node[below] {$m_3$} (x0);
\node[circle,draw, fill=black!30, minimum size = 0.2cm, inner sep=0pt] (y0) at (-2,0) {};
\node[circle,draw, fill=black!30, minimum size = 0.2cm, inner sep=0pt] (y1) at (-2,3) {};
\path[-stealth,thick]
(y1) edge node[left] {$m_1$} (y0)
(y0) edge[color = red, loop below] node {$m_2$} ();
\node[circle,draw, fill=black!30, minimum size = 0.2cm, inner sep=0pt] (z0) at (5,-1) {};
\node[circle,draw, fill=black!30, minimum size = 0.2cm, inner sep=0pt] (z1) at (7,0) {};
\path[-stealth,thick]
(z0) edge[color = red, loop left] node {$m_2$} ()
(z0) edge node[above] {$m_4$} (z1)
(z1) edge[color = blue, loop right] node {$m_3$} ();
\draw (-2,4) node {$K_1$};
\draw (0,6) node {$K_2$};
\draw (7,-1) node {$K_3$};
\draw (2,-1) node {$\mathbb{K}$};
\draw[color=red,fill = red!20]  (7.5,5) arc(180:45:0.5) -- (8.3536,5.3536) -- (10.3536,3 .3536) arc(45:-90:0.5) -- (10,2.5) -- (8,2.5) arc(270:180:0.5) -- (7.5,3) -- cycle;
\draw[color=blue,fill=blue!20]  (7.7172,4.7172) arc(225:45:0.4) -- (8.2828,5.2828) -- (10.2828,3 .2828) arc(45:-135:0.4) -- (9.7172,2.7172) -- cycle;
\draw[fill=black!20] (8,3) circle (0.3);
\draw[fill=black!20] (10,3) circle (0.3);
\draw (8,3) node {$K_1$};
\draw (8,5) node {$K_2$};
\draw (10,3) node {$K_3$};
\draw (8,3.5) node {$m_1$};
\draw (9.6,3.4) node {$m_4$};
\draw (9,4) node[color=blue] {$m_3$};
\draw (9,2.9) node[color=red] {$m_2$};
\draw (9,2) node {$\mathcal{H}_{\mathbb{K}}$};
\end{tikzpicture}
\caption{A product Kripke structure $\mathbb{K} = K_1 \mathbin{\parallel} K_2 \mathbin{\parallel} K_3$ and its synchronization hypergraph $\mathcal{H}_{\mathbb{K}}$. Shared actions are coloured; state labels and initial states are not depicted.} \label{fig:hyp}
\end{figure}
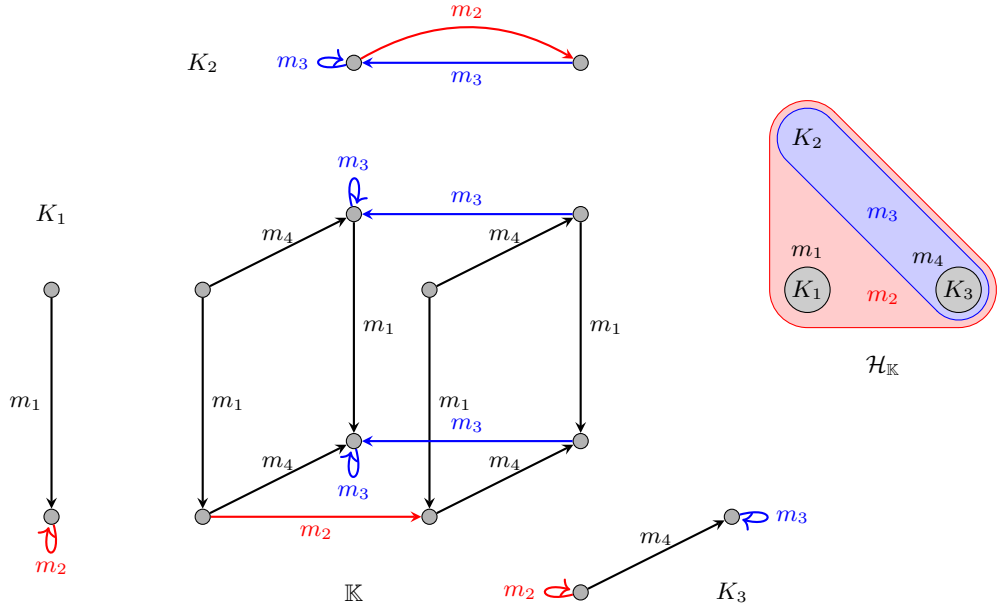

Large KSs will typically be obtained via the parallel composition of small ones. More precisely, we consider a large KS $\mathbb{K}$ with the following construction and notation:
\begin{itemize}
\item There are small \emph{basic KS} $K_1,\ldots,K_s$ with $\mathbb{K} = K_1 \mathbin{\parallel}\cdots \mathbin{\parallel}K_s$. As the $K_i$ are meant to be small, we take their size to be bounded independent of $s$; for simplicity, we assume that each $K_i$ has exactly $d$ states. As a result, we have $n = d^s$, and the system's growth can be described in $s$.
\item Each $K_i$ has action label set $M_i$; the action label set of $\mathbb{K}$ is $M = \bigcup_{i=1}^s M_i$.
\item For $m \in M$, we define $\mathcal{K}_m = \{K_i \mid m \in M_i\}$, the set of basic KS that have $m$ as an action label, and we set $s_m = |\mathcal{K}_m|$.
\item The interaction between the $K_i$ can be visualized by the \emph{synchronization hypergraph} $\mathcal{H}_{\mathbb{K}}$, whose vertex set is $\{K_1,\ldots,K_s\}$, and whose edge set is $\{\mathcal{K}_m \mid m \in M\}$; see Figure \ref{fig:hyp}.
\item For $m \in M$, we consider the KS $\mathbb{K}_m = \mathbin{\parallel}_{K_i \in \mathcal{K}_m} K_i$; it has $d^{s_m}$ states. Let $c_m$ be its number of $m$-labeled transitions, and let $f_m$ be its number of $m$-labeled self-loops.
\item In this terminology, the number of $m$-labeled transitions in $\mathbb{K}$ is exactly $d^{s-s_m}c_m$, so $\mathbb{K}$ has at most $\sum_{m \in M} d^{s-s_m}c_m$ transitions.
\item Let $\mathsf{av}(\mathbb{K})$ be the average degree in $\mathbb{K}$. Since $\mathbb{K}$ has $n = d^s$ states, it follows that $\mathsf{av}(\mathbb{K}) \leq \sum_{m \in M} \tfrac{c_m}{d^{s_m}}$.
\end{itemize}

We now want to understand the behaviour of $\mathbb{K}$ as $s$ (hence $n$) increases. In order to do so, we make a number of assumptions on the typical behaviour of large KS. These are as follows:

\begin{enumerate}
\item \emph{The synchronization hypergraph $\mathcal{H}_{\mathbb{K}}$ is connected.} This is because the basic KS of any connected component of $\mathcal{H}_{\mathbb{K}}$ represents a subsystem of $\mathbb{K}$ that acts completely independent of the other KS. In such cases, the KS representing each connected component is analyzed individually.
    \item \emph{There is a $C$ independent of $s$ such that $|M_i| \leq C$ for each $i$.} In order for the $K_i$ to be understandable by humans, there is a limit to their size and complexity. In particular, the total number of actions should be bounded.
    \item \emph{There is a $D$ independent of $s$ such that  $\tfrac{c_m}{d^{s_m}} \leq D$.} The fraction $\tfrac{c_m}{d^{s_m}}$ is the average degree when only considering $m$-labeled transitions, in the product KS consisting only of the basic KS partaking in $M$. To bound this is to bound the nondeterminism present in the action $m$ itself. This is not unreasonable; usually nondeterminism in KS comes from the choice of which action to take, instead of which transition to take for a given action. For instance, in KS derived from Petri nets each action (firing a transition) has a deterministic effect at each state where it applies, so $\tfrac{c_m}{d^{s_m}} \leq 1$ \cite{ma2016basis}. Note that this automatically holds if $s_m$ is bounded independent of $s$.
    \item \emph{There is a subset $M' \subseteq M$ with $|M'| = \Omega(s)$ such that $\tfrac{c_m}{d^{s_m}}$ is bounded away from $0$ for $m \in M'$.} In other words, there should be enough action labels with a decent amount of transitions. Since $c_m \geq 0$, it is enough to assume that there are many action labels $m$ for which $s_m$ is bounded independent of $s$. Thus it holds, for example, if $\Omega(s)$ of the $K_i$ have a `private' label for internal actions. This is not the case for KS from Petri nets, as these do not have strictly internal transitions. However, if there is a bound on each $s_m$ independent of $s$, then $|M|$ needs to be $\Omega(s)$ to ensure $\mathcal{H}_{\mathbb{K}}$ is connected; and under this assumption $\tfrac{c_m}{d^{s_m}} \geq d^{-s_m}$ which is bounded away from $0$.
    \item \emph{There is a $\varrho < 1$, independent of $s$, such that $\frac{f_m}{c_m} \leq \varrho$ for all $m \in M$.} This is a weak way of expressing `most $m$-labeled transitions are not self-loops'. This is justified by the modeling assumption that action labels are introduced to model nontrivial transitions. In KS derived from Petri nets we have $f_m = 0$ for all $m$, so this is automatically true.
\end{enumerate}

From these assumptions we determine the asymptotic behaviour of $\mathsf{av}(\mathbb{K})$:

\begin{lemma} \label{lem:logdeg}
If $\mathbb{K}$ satisfies assumptions 1-5 above, then $\mathsf{av}(\mathbb{K}) = \Theta(\log n)$.
\end{lemma}

\begin{proof}[Sketch]
From assumption 2 it follows that $|M| = \mathcal{O}(s)$; together with assumption 4 this means that $|M| = \Theta(s)$. Combining this with assumptions 3 and 4 we get $\sum_{m \in M} \tfrac{c_m}{d^{s_m}} = \Theta(s)$. Finally, one can use assumption 5 to show that $\mathsf{av}(\mathbb{K})$ differs from this by at most a constant factor, hence $\mathsf{av}(\mathbb{K}) = \Theta(s) = \Theta (\log n)$.
\end{proof}

We consider these assumptions justified for typical large KS; or at least, that deviations are small enough to not affect the asymptotic behaviour by more than a constant factor.  We can then conclude the following:

\begin{conclusion}
A typical large KS $\mathbb{K}$ has average degree $\propto \log n$. 
\end{conclusion}

\subsection{Empirical evaluation} \label{ssec:mcc}

\begin{figure}[t]
    \centering
    \begin{subfigure}[b]{0.49\textwidth}
    \centering
    \begin{tikzpicture}
		\tikzstyle{every node}=[font=\small]
		\begin{axis}[
        width = 60mm,
        height = 60mm,
        xlabel = {$\log(\#\text{states})$},
        xmin = 0,
        xmax = 1100,
        ylabel = {avg. degree},
        ylabel near ticks,
        axis lines = left,
        ]
		\addplot[x = log, y = deg, only marks, mark = +] table[col sep = comma] {scatter.csv};
        \addplot[mark=none, red] coordinates {(0,6.23) (1100,787.55)};
		\end{axis}
	\end{tikzpicture}
    \caption{all LTS}
    \end{subfigure}
    \begin{subfigure}[b]{0.49\textwidth}
    \centering
    \begin{tikzpicture}
		\tikzstyle{every node}=[font=\small]
		\begin{axis}[xmin = 0, xmax = 100, ymin = 0, ymax = 50,
        width = 60mm,
        height = 60mm,
        xlabel = {$\log(\#\text{states})$},
        ylabel = {avg. degree},
        ylabel near ticks,
        axis lines = left,]
		\addplot[x = log, y = deg, only marks, mark = +] table[col sep = comma] {scatter.csv};
        \addplot[mark=none, blue] coordinates {(0,8.6) (100,44.99)};
        \addplot[mark=none, red] coordinates {(0,6.23) (100,77.23)};
		\end{axis}
		\end{tikzpicture}
    \caption{avg.deg. $\leq 50$, \#states $\leq 10^{100}$.}
    \end{subfigure}
    \caption{Number of states (logarithmic) vs average degree of the 258 KS from the MCC benchmark for which these values are published. (b) zooms in on KS with average degree $\leq 50$ and at most $10^{100}$ states (246 out of 281). Red line: linear regression based on all KS; Blue line: linear regression based only on small KS.} \label{fig:MCC}
\end{figure}
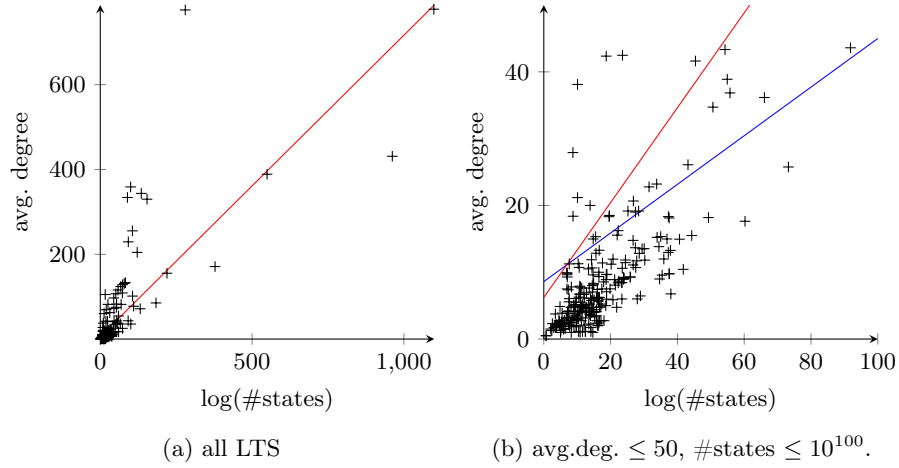

To confirm Conclusion 1 empirically, we consider the \emph{Model Checking Competition} \cite{amat2023behind}, an annual competition for model checking concurrent systems. The benchmark used in the competition consists of Petri nets, but the object evaluated in model checking are their underlying KS of reachable states. We consider all MCC Petri nets for which, up to the 2026 competition, the size of their KS has been computed (number of states and transitions): this results in a set of 281 KS (263 academic, 18 industrial) of size $2$--$10^{254}$.

For these KS, we plot $\tfrac{m}{n}$ vs $\log n$, where $n$ is the number of states and $m$ the number of transitions. Conclusion 1 states that these should be linearly related. The resulting plots are shown in Fig.~\ref{fig:MCC}. Fig.~\ref{fig:MCC}(a) is hard to read, even on a log scale, because of several large outliers. If we zoom in on the LTS with average degree $\leq 50$ and at most $10^{100}$ states (246 out of 281) we get a clearer picture, in which a linear relation seems reasonable. Unfortunately, linear regression does not give a good predictor of the relation, since the outliers upwards are much larger than the outliers downwards, and the trend line ends up above the majority of points.

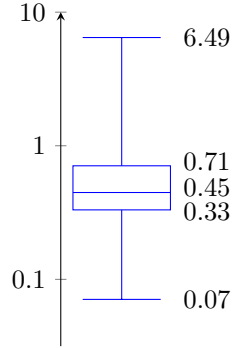
\begin{wrapfigure}[14]{right}{0.3\linewidth}
\centering
 \vspace{-3em}
\begin{tikzpicture}
\begin{axis}[boxplot/draw direction = y,
axis x line = none,
axis y line = left,
xmin = 0.5,
xmax = 2,
ymin = -1.5,
ymax = 1,
width = 40mm,
height = 60mm,
ytick = {-1,0,1},
yticklabels = {$0.1$,$1$,$10$}
]
    \addplot+ [
        boxplot prepared={
            lower whisker=-1.15, lower quartile=-0.48,
            median=-0.35,
            upper quartile=-0.15, upper whisker=0.81,
        },
    ] coordinates {};
    \draw (axis cs:1.7,-1.15) node {$0.07$};
    \draw (axis cs:1.7,-0.49) node {$0.33$};
    \draw (axis cs:1.7,-0.31) node {$0.45$};
    \draw (axis cs:1.7,-0.12) node {$0.71$};
    \draw (axis cs:1.7,0.81) node {$6.49$};
\end{axis}
\end{tikzpicture}
\caption{Boxplot of values of $\frac{m}{n\log n}$ for MCC KS ($n$ states, $m$ transitions).} \label{fig:boxplot}
\end{wrapfigure}
More instructive is the boxplot of $\tfrac{m}{n \log n}$ of Fig.~\ref{fig:boxplot}. If our hypothesis holds, this value should be approximately constant. Indeed, we find that the first and third quartile only differ by a factor 2.15, and all values are within a factor 100 of each other. While this is not perfect, it shows that logarithmically growing degrees provide a considerably more realistic model than Erd\H{o}s-Rényi models with constant probability \cite{dong20250}: there one expects $\tfrac{m}{n^2}$ to be constant, but this takes values $0.25$--$10^{-1091}$. This is a considerably worse spread, and we conclude that taking the average degree in a large KS to be proportional to $\log n$ is a reasonable assumption.

\subsection{Number of initial states}

By Definition \ref{def:parco}, a state in the composition $\mathbb{K}$ is initial if and only if it is initial in each basic KS $K_i$. Hence, if we assume that in each $K_i$ we have $c$ initial states, the total amount of initial states is
\[
c^s = c^{\log_d n} = n^{\log_d c} := n^{1-r},
\]
where $0 \leq r \leq 1$ because $1 \leq c \leq d$.

\begin{conclusion}
There exists a $0 \leq r \leq 1$ such that a typical large KS $\mathbb{K}$ has $\approx n^{1-r}$ initial states.
\end{conclusion}

\subsection{Atomic propositions} \label{ssec:ap}

Finally, we need to understand the asymptotic behaviour of atomic propositions, i.e., as $n \rightarrow \infty$, how many states satisfy each atomic proposition. If $K_i$ has set of atomic propositions $A_i$, we see from Definition \ref{def:parco} that $\mathbb{K}$ has set of atomic propositions $\bigcup_{i=1}^s A_i$. We assume that the $A_i$ are disjoint; this holds in many cases, such as when $\mathbb{K}$ comes from a Petri net. Hence if $a \in A_i$, then the state $v = (v_1,\ldots,v_s)$ of $\mathbb{K}$ satisfies $a$ if and only if $v_i$ satisfies $a$ in $K_i$. If there are $c$ such states in $K_i$, it follows that there are $d^{s-1}c = \tfrac{c}{d}n$ satisfying states in $\mathbb{K}$. We conclude:

\begin{conclusion}
In a typical large KS $\mathbb{K}$, the number of states satisfying a given atomic proposition is proportional to $n$.
\end{conclusion}

\section{Random Kripke structures} \label{sec:prelim}

In this section, we define the random KS model that we will consider in this paper, based on the conclusions of Section \ref{sec:prob}. We take an Erd\H{o}s-Rényi style approach: all events of the form `a transition exists', `a state is initial' and `a state satisfies an atomic proposition' are considered independent. This is a simplification, but this model needs to be understood first before it makes sense to study more intricate probabilistic models that also model the relation between events. Furthermore, in practice studying models with this degree of independence is sufficient for applications in statistical physics \cite{andreanov2004large} and computer networks \cite{simonian1995large}.

Given such an independent model, the governing parameters are the probabilities of a transition being present, a state being initial, and a state satisfying an atomic proposition. Based on Conclusions 1-3, we take the first of these to scale as $\tfrac{\log n}{n}$, the second as $n^{-r}$, and the third as constant.

\begin{definition}
Let $(\vec{p},r,c) \in (0,1)^{A} \times [0,1) \times \mathbb{R}_{>1}$, and let $n \in \mathbb{Z}_{\geq 1}$. Then the \emph{random labeled digraph $\mathbb{D}_{n,\vec{p},\alpha,c} = (V,E,I,\ell)$ is defined as follows:}
\begin{enumerate}
\item $V = \{1,\ldots,n\}$;
\item Each transition $x \rightarrow y$ exists with probability $p_{\recht{t}} := \tfrac{c\log(n)}{n}$;
\item Each state $x$ is initial with probability $p_{\recht{i}} := n^{-r}$;
\item For each state $x$ and each label $a$, we have $a \in \ell(x)$ with probability $p_a$;
\item All these events are independent.
\end{enumerate}
\end{definition}

Note that compared to Section \ref{sec:prob}, we demand $r < 1$ rather than $r \leq 1$. As we will see below, this is needed to rule out the situation where there would be a nonzero probability of having no initial states. Taking $r = 1$  corresponds to only having a single initial state; we also consider such a model, but with this condition strictly enforced, see Def.~\ref{def:K1}. Finally, we demand $c > 1$ to ensure the existence of outgoing transitions, see Theorem \ref{thm:KS}.

\begin{remark} \label{rem:fixA}
In Section \ref{ssec:ap} we discussed that atomic propositions reflect properties of the basic KS. Since the number of these basic KS goes to infinity, it would make sense to consider infinite sets $A$. However, as we will see in Sections \ref{sec:ltl} and \ref{sec:ctl}, the logic formulas we will study will only have a finite number of variables, and we will only study the behaviour of one formula at a time, so taking $A$ finite is not a restriction.
\end{remark}

\begin{definition}
The \emph{random Kripke structure} $\mathbb{K}_{n,\vec{p},c}$ is defined to be $\mathbb{D}_{n,\vec{p},r,c}$ conditioned on the fact that it is a Kripke structure; for every Kripke structure with state set $V = \{1,\ldots,n\}$ we have
\[
\mathbb{P}(\mathbb{K}_{n,\vec{p},r,c} = K) = \frac{\mathbb{P}(\mathbb{D}_{n,\vec{p},r,c} = K)}{\mathbb{P}(\mathbb{D}_{n,\vec{p},r,c} \textrm{ is a Kripke structure})}.
\]
\end{definition}

Throughout this paper, $\vec{p}$, $r$ and $c$ are generally fixed, while $n$ often goes to infinity; we will omit subscripts when no confusion can arise.

The following result shows that asymptotically, the difference between $\mathbb{K}$ and $\mathbb{D}$ is immaterial:

\begin{theorem} \label{thm:KS}
As $n \rightarrow \infty$ we have $\mathbb{P}(\mathbb{D} \textrm{ is a Kripke structure}) \rightarrow 1$. 
\end{theorem}

\begin{proof}[Sketch]
We have to show that 
\begin{align*}
\mathbb{P}(I \neq \varnothing) &\rightarrow 1,\\
\mathbb{P}(\textrm{each state has an outgoing transition}) &\rightarrow 1,
\end{align*}
as these events are independent. The first probability is $1-(1-n^{-r})^n$, which one can work out to converge to 1. The second one is purely a statement about random directed graphs and has been proven in \cite{graham2008note}; the assumption $c > 1$ plays an essential role here.
\end{proof}

We will make frequent use of this fact, by studying $\mathbb{K}$ using the (easier) probability distribution of $\mathbb{D}$. 

\begin{remark}
In fact, $\mathbb{D}$ (hence also $\mathbb{K}$) is strongly connected with high probability.
\end{remark}

We also consider the situation where there is just a single initial state:

\begin{definition} \label{def:K1}
The random KS $\mathbb{K}^1$ is defined to be $\mathbb{K}$ conditioned on the fact that $|I| = 1$. The single initial state is denoted $\ini$.
\end{definition}

In the next two chapters, we study the behaviour of $\mathbb{K}$ and $\mathbb{K}^1$ under two of the most prominent logics used in model checking: LTL and CTL.

\section{Behaviour under LTL} \label{sec:ltl}

In this chapter, we study the behaviour of our random KSs under \emph{Linear Temporal Logic} (LTL) \cite{pnueli1977temporal}, a logic that describes the behaviour of runs of the program (infinite paths in the directed graph). We first briefly recap the syntax and semantics of LTL before going to our results.

\subsection{Syntax}

The syntax of LTL formulas $\varphi$ is given by the following grammar:

\[\varphi ::= a\mid \varphi \wedge \varphi \mid \neg \varphi \mid \mathsf{X} \varphi\mid \mathsf{F}\varphi \mid \mathsf{G} \varphi \mid \varphi \mathsf{U}\varphi\]

Where $a \in A$. We define Boolean operators $\vee,\rightarrow,...$ as usual from $\wedge$ and $\neg$.

\subsection{Semantics}

A \emph{path} in a KS $K$ is an infinite sequence of states $\pi = (\pi_i)_{i \geq 0}$ of $K$,  such that $\pi_i\pi_{i+1}$ is a transition in $K$ for all $i \geq 0$. Such a path represents a possible execution of a program. LTL formulas are path formulas: they say something about paths in a KS. It will be convenient to also define the validity of infinite words $w = w_0w_1w_2\cdots \in (2^A)^{\omega}$ of subsets of $A$. For such a $w$, and for a KS $K$ and a path $\pi$ in $K$, we define

\begin{align*}
w &\models a &\Leftrightarrow&& a &\in w_0\\
w &\models \varphi \wedge \varphi' &\Leftrightarrow&& w &\models \varphi \textrm{ and } w \models \varphi'\\
w &\models \neg \varphi &\Leftrightarrow&& w &\not \models \varphi\\
w &\models \mathsf{X} \varphi &\Leftrightarrow&& (w_j)_{j \geq 1} &\models \varphi\\
w &\models \mathsf{F} \varphi &\Leftrightarrow&& (w_j)_{j \geq i} &\models \varphi \textrm{ for some $i \geq 0$}\\
w &\models \mathsf{G} \varphi &\Leftrightarrow&&  (w_j)_{j \geq i} &\models \varphi \textrm{ for all $i \geq 0$}\\
w &\models \varphi \mathsf{U} \varphi' &\Leftrightarrow&& (w_j)_{j \geq i} &\models \varphi' \textrm{ for some $i$ and } (w_j)_{j \geq i'} \models \varphi \textrm{ for all $i' < i$}\\
K,\pi &\models \varphi & \Leftrightarrow&& \ell(\pi_0)\ell(\pi_1)\cdots &\models \varphi\\
K &\models \varphi & \Leftrightarrow&& K,\pi &\models \varphi \textrm{ for all paths $\pi$ starting in an initial state.}
\end{align*}

Here we recall the fact that we assigned a set $\ell(v) \subseteq A$ to each state $v$. In these semantics, $\ell(v) \subseteq A$ is the set of atomic propositions satisfied by a state $v$ (i.e., by all paths starting in $v$). For what follows, it will be useful to define specific `propositional' formulas, i.e., formulas with only Boolean operators and no temporal operators $\mathsf{X},\mathsf{F},\mathsf{G},\mathsf{U}$. For $S \subseteq A$, we consider the formula $\alpha_S := \bigwedge_{a \in S} a \wedge \bigwedge_{a \in A \setminus S} \neg a$; hence $K,\pi \models \alpha_S$ if and only if $\ell(\pi_0) = S$. The probability of $\ell(v) = a$ in $\mathbb{K}$ is given by
\begin{equation}
\varrho_S := \prod_{a \in S} p_a \cdot \prod_{a \in A \setminus S} (1-p_a) \in [0,1]. \label{eq:rho}
\end{equation}

We define tautologies to be formulas that are true in all KSs. It will also be convenient to define tautologies relative to an $S \subseteq A$:

\begin{definition}
An LTL formula $\varphi$ is called a \emph{tautology} if $K,\pi \models \varphi$ for all paths $\pi$ in all KSs $K$. For $S \in 2^A$, we say that $\varphi$ is an \emph{$S$-tautology} if $\alpha_S \rightarrow \varphi$ is a tautology.
\end{definition}

If $\varphi$ is not a tautology then there exists a `lasso'-shaped counterexample (an infinite loop preceded by an initial segment):

\begin{theorem} \label{thm:lasso}
\emph{\cite{vardi1986automata}} $\varphi$ is not a tautology if and only if there exist $u,v \in (2^{A})^*$ such that $uv^{\omega} \not \models \varphi$.
\end{theorem}

For future purpose, we also state the following two results on $S$-tautologies. The first statement is obvious but nevertheless useful, the second one follows immediately from Theorem \ref{thm:lasso}.

\begin{corollary} \label{cor:counterword}
Let $\varphi$ be an LTL formula.
\begin{enumerate}
\item $\varphi$ is a tautology if and only if it is an $S$-tautology for each $S \in 2^A$.
\item Let $S \in 2^A$. Then $\varphi$ is not an $S$-tautology if and only if there exist $u,v \in (2^A)^*$ such that $uv^{\omega} \not \models \varphi$ and $|u| > 0$ and $u_1 = S$.
\end{enumerate}
\end{corollary}

\begin{example} \label{ex:taut}
Let $A = \{a_1,a_2\}$, and consider $\varphi = a_2 \wedge \mathsf{F}a_1$. An infinite word $w = w_0w_1\cdots \in (2^A)^{\omega}$ satisfies $\varphi$ if and only if $a_2 \in w_0$ and there is an $i$ such that $a \in w_i$. Then $\varphi$ is an $\{a_1,a_2\}$-tautology: Any $w$ with $w_0 = \{a_1,a_2\}$ satisfies $\varphi$, so $(a_1 \wedge a_2) \rightarrow \varphi$ is a tautology. On the other hand, $\varphi$ is not an $S$-tautology for any other $S \in 2^A$, as we have
\[
\varnothing^{\omega} \not \models \varphi, \quad \quad \{a_1\}^{\omega} \not \models \varphi, \quad \quad \{a_2\}^{\omega} \not \models \varphi;
\]
this shows that for every other $S$, we have a lasso starting in $S$ that does not satisfy $\varphi$.
\end{example}

We end these preliminaries with the complexity of LTL tautology checking:

\begin{theorem} \emph{\cite{sistla1985complexity}} \label{thm:ltlcomp}
The problem of determining whether a given LTL formula is a tautology, is PSPACE-complete.
\end{theorem}

LTL tautology checking (or equivalently, satisfiability checking) is a well-studied problem in model checking, with many available techniques \cite{rozier2007ltl,li2014aalta,geatti2024sat}.

\subsection{0-1 laws}

In this section we discuss the convergence behaviour of $\mathbb{P}(\mathbb{K} \models \varphi)$ and $\mathbb{P}(\mathbb{K}^1 \models \varphi)$, for LTL formulas $\varphi$. The main results are Theorem \ref{thm:01ltl} and \ref{thm:01ltlone}, in which we show that these converge for all $\varphi$; and furthermore that the limit can explicitly be determined. These two results both rely on the following result:

\begin{theorem} \label{thm:01ltlmain}
Let $\varphi$ be an LTL formula, and let $S \in 2^A$. Then
\[
\lim_{n \rightarrow \infty} \mathbb{P}(\mathbb{K}^1 \models \varphi \mid \ell(\ini) = S) = \begin{cases}
1, & \textrm{ if $\varphi$ is an $S$-tautology,}\\
0, & \textrm{ otherwise}.
\end{cases}
\]
\end{theorem}

\begin{proof}[Sketch]
The interesting case is to show that the limit is $0$ when $\varphi$ is not an $S$-tautology. To do so, consider an infinite word $uv^{\omega}$ that is a counterexample for $\varphi$, with $u_0 = S$, and let $T \subseteq \mathbb{K}^1$ be the set of vertices labeled $v_1$. We define the random graph $\hat{T}$ to have vertex set $T$, and the edge $x \rightarrow y$ exists in $\hat{T}$ iff there is a $v$-labeled path from $x$ to $y$. One can prove that asymptotically, $\hat{T}$ behaves like a directed Erd\H{o}s-Rényi graph with edge probability $\Theta((\log n)^{|v|}n^{-1})$. This is $\gg n^{-1}$, so it has a strongly connected component consisting of nearly all vertices in $T$. Within this strongly connected component, one can keep taking $v$-labeled paths indefinitely. Furthermore, this component is large enough that a $u$-labeled path from $\ini$ into it exists almost surely; if this is the case, then the counterexample path $uv^{\omega}$ exists in $\mathbb{K}^1$, so $\mathbb{P}(\mathbb{K}^1 \models \varphi \mid \ell(\ini) = S) \rightarrow 0$.
\end{proof}

From this theorem we derive our two main results on LTL. The first one states that for every formula $\varphi$, the probability $\mathbb{P}(\mathbb{K}\models \varphi)$ converges to either 0 or 1; in random graph theory, such a result is called a \emph{0-1 law}.

\begin{theorem}[0-1 law for LTL on $\mathbb{K}$] \label{thm:01ltl}
Let $\varphi$ be an LTL formula. Then $\mathbb{P}(\mathbb{K} \models \varphi) \rightarrow 1$ if $\varphi$ is a tautology, and $\mathbb{P}(\mathbb{K} \models \varphi) \rightarrow 0$ otherwise.
\end{theorem}

\begin{proof}[Sketch]
If $\varphi$ is not a tautology, then there exists an $S$ such that it is not an $S$-tautology. With high probability, there exists an $x \in I$ such that $\ell(x) = S$. Shrinking $I$  can only increase $\mathbb{P}(\mathbb{K} \models \varphi)$; however, if we set $I = \{x\}$, Theorem \ref{thm:01ltlmain} tells us that $\mathbb{P}(\mathbb{K} \models \varphi) \rightarrow 0$.
\end{proof}

Our second main result is more subtle. For $\mathbb{K}^1$, we cannot expect a 0-1 law to hold: $\mathbb{K}^1 \models a$ iff $a \in \ell(\ini)$, which happens with probability $p_a$. Instead, we show that LTL ion $\mathbb{K}^1$ satisfies a \emph{convergence law}, i.e., the truth probability of every formula converges:

\begin{theorem} \label{thm:01ltlone}
Let $\varphi$ be an LTL formula, and let $\mathcal{T}_{\varphi} \subseteq 2^A$ be the set of $S$ such that $\varphi$ is an $S$-tautology. Let $\varrho_S$ be as in \eqref{eq:rho}. Then
\[
\lim_{n \rightarrow \infty} \mathbb{P}(\mathbb{K}^1 \models \varphi) = \sum_{S \in \mathcal{T}_{\varphi}} \varrho_S.
\]
\end{theorem}

\begin{proof}
Using Theorem \ref{thm:01ltlmain}, we have
\begin{align*}
\lim_{n \rightarrow \infty} \mathbb{P}(\mathbb{K}^1 \models \varphi) &= \sum_{S \in 2^A} \mathbb{P}(\mathbb{K}^1 \models \varphi \mid \ell(\ini) = S)\mathbb{P}(\ell(\ini) = S) \\
&= \sum_{S \in \mathcal{T}_{\varphi}} \varrho_S.
\end{align*}
\end{proof}

\begin{example}
Continuing Example \ref{ex:taut}, we have $\mathcal{T}_{\varphi} = \{\{a_1,a_2\}\}$; it follows that $\lim_{n \rightarrow \infty} \mathbb{P}(\mathbb{K}^1 \models \varphi) = \varrho_{\{a_1,a_2\}} = p_{a_1}p_{a_2}$.
\end{example}

\begin{remark}
Our results are similar to the LTL results of \cite{dong20250}, which studied Erd\H{o}s-Rényi type KSs where all probabilities are constant rather than logarithmic $p_{\recht{t}}$ and power-law $p_{\recht{i}}$. Interestingly, LTL formulas behave exactly the same in that model, i.e, every LTL formula has the same limit probability in that model as it has in our $\mathbb{K}$ and $\mathbb{K}^1$. Our proofs are significantly more involved, though: the lynchpin of the results of \cite{dong20250} is the \emph{extension principle}, which roughly states that every KS $K$ can be found in $\mathbb{K}$ if $n \rightarrow \infty$. This means that for every non-tautology, one can find a counterexample KS in $\mathbb{K}$ with high probability. However, this extension principle does not hold in our case: because $p_{\recht{t}} \rightarrow 0$ quickly, we cannot expect very dense $K$ to be found in $\mathbb{K}$. Instead, our approach relies on arguing that counterexample \emph{paths} can always be found: because these are not dense (as graphs), this problem is circumvented. Arguing the existence of paths is substantially more involved, however.
\end{remark}

\begin{remark}
A closer look at our proofs show that they also hold for more general choices of the functions $p_{\recht{t}}$ and $p_{\recht{i}}$: we only use lower bounds for both, so the result holds whenever $\liminf \tfrac{p_{\recht{t}}(n)}{\log n} > 1$ and $p_{\recht{i}} \gg n^{-1}$. Hence our results generalize \cite{dong20250}. It would be interesting to further explore to what extent LTL behaviour depends on parameter scaling, and in particular the form of the \emph{threshold functions} of individual LTL formulas $\varphi$: A function $t\colon \mathbb{Z}_{\geq 1} \rightarrow \mathbb{R}_{>0}$ is a function is said to be a threshold for $\varphi$ if $\mathbb{P}(\mathbb{K}\models \varphi) \rightarrow 1$ whenever $p_{\recht{t}}(n) \ll t(n)$ as $n \rightarrow \infty$, and $\varphi$ if $\mathbb{P}(\mathbb{K}\models \varphi) \rightarrow 0$ whenever $p_{\recht{t}}(n) \gg t(n)$ (we fix the behaviour of $p_{\recht{i}}$ and the $p_a$) \cite{bollobas1987threshold}. In other words, the behaviour of $\mathbb{K}$ w.r.t. $\varphi$ undergoes a phase transition around $t$. Since satisfying $\varphi$ is a monotonically decreasing property (adding transitions can only make $\varphi$ more false, never more true), we know by \cite{bollobas1987threshold} that every $\varphi$ has such a threshold; investigating the form of these threshold is interesting work for future research.
\end{remark}

\subsection{Computing limit behaviour}

In fact, our result not only shows the existence of a 0-1 law, but also how to decide between 0 and 1. Combining Theorem \ref{thm:01ltl} and Theorem \ref{thm:ltlcomp} yields the following result:

\begin{corollary}
The problem of deciding $\lim_{n \rightarrow \infty} \mathbb{P}(\mathbb{K} \models \varphi)$ is PSPACE-complete.
\end{corollary}

Because deciding this limit is equivalent to tautology checking, established LTL satisfiability checking tools \cite{rozier2007ltl,li2014aalta,geatti2024sat} can directly be used for determining the limit behaviour of LTL formulas in $\mathbb{K}$. The following result shows that the situation for $\mathbb{K}^1$ is similar:

\begin{corollary}
Computing $\lim_{n\rightarrow \infty} \mathbb{P}(\mathbb{K}^1 \models \varphi)$ is NP-hard.
\end{corollary}

\begin{proof}
Suppose that $0 < p_a < 1$ for all $a \in A$. Then $\varrho_S > 1$ for all $S \in 2^A$; hence for any LTL formula $\varphi$ we have
\begin{align*}
\varphi \textrm{ is a tautology} \ \  &\Leftrightarrow \varphi \textrm{ is an $S$-tautology for all $S \in 2^A$}\\
&\Leftrightarrow \lim_{n \rightarrow \infty} \mathbb{P}(\mathbb{K}^1 \models \varphi) = 1.
\end{align*}
By Theorem \ref{thm:ltlcomp} this shows that computing $\mathbb{P}(\mathbb{K}^1 \models \varphi)$ is NP-hard.
\end{proof}

\section{CTL} \label{sec:ctl}

In this chapter, we study the behaviour of \emph{Computational Tree Logic} (CTL) \cite{baier2008principles}. We follow the same structure as Section \ref{sec:ltl}.

\subsection{Syntax}

Contrary to LTL, CTL consists of \emph{state formulas} $\psi$ and \emph{path formulas $\varphi$}:

\begin{align*}
\psi &::= a \mid \psi \wedge \psi \mid \neg \psi \mid \mathsf{A}\varphi \mid \mathsf{E}\varphi \\
\varphi &::= \mathsf{X}\psi \mid \mathsf{F}\psi \mid \mathsf{G}\psi \mid \psi \mathsf{U}\psi.
\end{align*}

CTL model checking is only concerned with state formulas, as we will see below. Therefore, one can alternatively skip $\varphi$ altogether and instead define the syntax in terms of $\mathsf{AX}, \mathsf{EU},...$ directly.

\subsection{Semantics}

For a KS $K$ with state $x$ and path $\pi$, we define $K,x \models \psi$ and $K,\pi \models \varphi$ as follows:

\begin{align*}
K,x &\models a &\Leftrightarrow&& a &\in \ell(x) \\
K,x &\models \psi \wedge \psi' &\Leftrightarrow&& K,x &\models \psi \textrm{ and } K,x \models \psi'\\
K,x &\models \neg \psi &\Leftrightarrow&& K,x &\not\models \psi\\
K,x &\models \mathsf{A}\varphi &\Leftrightarrow&& K,\pi &\models \varphi \textrm{ for all paths $\pi$ with $\pi_0 = x$}\\
K,x &\models \mathsf{E}\varphi &\Leftrightarrow&& K,\pi &\models \varphi \textrm{ for some path $\pi$ with $\pi_0 = x$}\\
K,\pi &\models \mathsf{X}\psi &\Leftrightarrow&& K,\pi_1 &\models \psi \\
K,\pi &\models \mathsf{F}\psi &\Leftrightarrow&& K,\pi_i &\models \psi \textrm{ for some $i \geq 0$}\\
K,\pi &\models \mathsf{G}\psi &\Leftrightarrow&& K,\pi_i &\models \psi \textrm{ for all $i \geq 0$}\\
K,\pi &\models \psi \mathsf{U}\psi' &\Leftrightarrow&& K,\pi_i &\models \psi' \textrm{ for some $i \geq 0$, and } K,\pi_{j} \models \psi \textrm{ for all $j < i$}\\
K &\models \psi & \Leftrightarrow&& K,x &\models \psi \textrm{ for all initial states $x$}.
\end{align*}

The name \emph{computational tree logic} refers to the fact that a CTL formula does not look at a single path, but at the entire infinite tree of possible paths from a given state; branching is explored using $\mathsf{A}$ (all outgoing paths) and $\mathsf{E}$ (a single path).

CTL does not contain LTL, nor the other way around: the CTL path formula $\mathsf{AFAG}a$ cannot be expressed in LTL, and the LTL formula $\mathsf{FG}a$ cannot be expressed in CTL \cite{baier2008principles}. An example of a logic extending both is CTL*.

As for LTL, we consider propositional formulas $\alpha$ that only have Boolean connectives and no temporal operators. We write $\alpha \equiv \alpha'$ if two propositional formulas are equivalent. Furthermore,  we define 
\begin{align}
\mathcal{S}_{\alpha} &= \{S \in 2^A \mid (\alpha_S \rightarrow \alpha) \equiv \top\}, \label{eq:Salpha}\\
\varrho_{\alpha} &= \sum_{S \in \mathcal{S}_{\alpha}} \varrho_S \in [0,1]. \nonumber
\end{align}
Hence $\mathcal{S}_{\alpha}$ is the set of all $S$ under which $\alpha$ holds, and $\mathbb{P}(\mathbb{K},x \models \alpha) = \varrho_{\alpha}$ for every $x \in V$. Note that $\varrho_{\alpha_S} = \varrho_S$. 

\subsection{0-1 laws}

Unlike for LTL, for CTL there are non-tautologies that nevertheless have limit probability $1$. In fact, the limit probability may depend on the transition probability parameter $p_{\recht{t}} = \tfrac{c \log n}{n}$ and on the initial state probability $p_{\recht{i}} = n^{-r}$, as the following example shows.

\begin{example} \label{ex:ctl}
Consider $\varphi = \mathsf{EX}a$. Then $\mathbb{K} \models \varphi$ if and only if every initial state has a successor state satisfying $a$; it is clearly not a tautology. For any given state, the probability that a given other state is an $a$-successor is $p_ap_{\recht{t}}$. Let $C_x$ be the indicator of the event that state $x$ is initial and has no $a$-successor; then $\mathbb{P}(C_x = 1) = p_{\recht{i}}(1-p_ap_{\recht{t}})^n$. If $X = \sum_x C_x$, then $\mathbb{E}[X] = np_{\recht{i}}(1-p_ap_{\recht{t}})^n$. One can show that this is $(1+o(1))n^{1-r-p_ac}$ (recall that $p_{\recht{i}} = n^{-r}$). Thus, the behaviour of $\varphi$ depends on the value of $r+p_ac$ (see Appendix \ref{app:ex} for details):
\begin{itemize}
\item If $r+p_ac > 1$, then $\mathbb{E}[X] \rightarrow 0$, and $\mathbb{P}(K \models \varphi) = \mathbb{P}(X = 0) \rightarrow 1$. Hence $\varphi$ is an example of a non-tautology with limit probability 1.
\item If $r+p_ac < 1$, then $\mathbb{E}[X] \rightarrow \infty$. With the second moment method one can show that $\mathbb{P}(X= 0) \rightarrow 0$, so $\mathbb{P}(\mathbb{K} \models \varphi) \rightarrow 0$.
\item If $r+p_ac = 1$, then $\mathbb{E}[X]$ converges to the constant $p_{\recht{i}}$. In fact, one can show that $X$ converges to a Poisson random variable with parameter $1$, so we get $\mathbb{P}(X = 0) \rightarrow \textrm{e}^{-1}$. This shows that for some values of $(\vec{p},r,c)$ a 0-1 law fails.
\end{itemize}
\end{example}

The moral of the story is that unlike for LTL, the limit probabilities of certain CTL formulas depend on the parameters $(\vec{p},r,c)$. The problem is that for values of $c > 1$ that are small, we get connectivity as every state has a successor, but not every state will have an $S$-labeled successor, for certain $S \in 2^A$. When this complication does not happen, we indeed get a 0-1 law for multiple initial states, and a convergence law for a single initial state:

\begin{theorem} \label{thm:ctl}
Suppose that $c > \max_{S \in 2^A} \varrho_S^{-1}$. Then $\mathbb{K}$ and satisfies a 0-1 law, and $\mathbb{K}^1$ satisfies a convergence law. Moreover, the limit probability of $\psi$ can be computed in time $\mathcal{O}(2^{|A|}|\psi|)$, where $|\psi|$ denotes formula length.
\end{theorem}

This result is analogous to our result for CTL on Erd\H{o}s-Rényi KSs with constant transition probability \cite{dong20250}, and the proof structure is the same as well; nevertheless, the details require considerably more careful probabilistic reasoning, since the extension property fails for our sparse KS. We sketch the proof strategy below, and also explain why the case of smaller $c$ would require a radically different approach. We will leave finding such an approach to future work.

\begin{remark}
If one considers $A$ fixed, then Theorem \ref{thm:ctl} shows that 0-1 law checking for CTL has polynomial time complexity. However, we state the result like this to account for Remark \ref{rem:fixA}, where one might want $A$ to differ between formulas. From that perspective, the complexity can be expressed as $\mathcal{O}(2^{|\psi|}|\psi|)$. The exponential dependence on $|\psi|$ may seem surprising compared to regular CTL model checking: determining whether $K \models \varphi$ for a concrete KS $K$ has complexity $\mathcal{O}(|K|\cdot|\varphi|)$. This is best understood by the idea that the dependence on $K$ (does $\varphi$ hold for $K$?) is at some points replaced by satisfiability checking (can $\varphi$ hold at all?); see Lemma \ref{lem:bustep}. This satisfiability checking is only done on propositional $\varphi$, which has complexity $\mathcal{O}(2^{|A|})$.
\end{remark}

\subsection{Proving Theorem \ref{thm:ctl}.}

The core idea is to define an equivalence relation $\sim$ on CTL formulas such that each formula is equivalent to a propositional one (only Boolean operators), and such that equivalent formulas have the same limit probability. This way, we reduce the problem to propositional formulas, which are easily understood.

\begin{definition}
For a CTL formula $\varphi$, let $\sem{\varphi} = \{x\in \mathbb{K} \mid \mathbb{K},x \models \varphi\}$ be the set of vertices in $\mathbb{K}$ that satisfy $\varphi$. For two CTL formulas, we write $\varphi \sim \varphi'$ if $\mathbb{P}(\sem{\varphi} = \sem{\varphi'}) \rightarrow 1$.
\end{definition}

Clearly, if $\varphi \sim \varphi'$, then $\varphi$ and $\varphi'$ have the same limit probability, both in $\mathbb{K}$ and $\mathbb{K}^1$. Since $\sem{\varphi}$ can be computed bottom-up from the syntax tree of $\varphi$, we get the following result:

\begin{lemma}[See \cite{dong20250}] \label{lem:rep}
Let $\varphi,\varphi'$ be two CTL formulas, and let $\psi$ be a CTL formula of which $\varphi$ is a subformula. If $\varphi \sim \varphi'$, then $\psi \sim \psi[\varphi \leftarrow \varphi']$.
\end{lemma}

The following result shows that we can eliminate temporal operators ($\mathsf{EX}$ etc.) in CTL formulas that have only one of them:

\begin{lemma} \label{lem:bustep}
Suppose that $c > \max_{S \in 2^A} \varrho_S^{-1}$. Let $\alpha$ and $\alpha'$ be propositional CTL formulas, i.e., formulas with only Boolean connectives. then
\begin{align*}
\mathsf{EX}\alpha \sim \mathsf{EF}\alpha &\sim \begin{cases}
\bot, & \textrm{ if $\alpha \equiv \bot$},\\
\top, & \textrm{ otherwise}
\end{cases}\\
\mathsf{AX} \alpha \sim \mathsf{AG}\alpha &\sim \begin{cases}
\top, & \textrm{ if $\alpha \equiv \top$},\\
\bot, & \textrm{ otherwise}
\end{cases}\\
\mathsf{AF}\alpha \sim \mathsf{EG}\alpha \sim \mathsf{A}(\alpha'\mathsf{U}\alpha) &\sim \alpha,\\
\mathsf{E}(\alpha'\mathsf{U}\alpha) &\sim \begin{cases}
\bot, & \textrm{ if $\alpha \equiv \bot$},\\
\alpha \vee \alpha', & \textrm{ otherwise.}
\end{cases}
\end{align*}
\end{lemma}

\begin{proof}[Sketch of proof]
The key cases are $\mathsf{EX}\alpha$, essentially covered in Example \ref{ex:ctl}, and $\mathsf{A}(\alpha'\mathsf{U}\alpha)$, which we discuss here. If $\alpha \not \equiv \top$, then one can show that our choice of $c$ implies that $\llbracket \neg \alpha \rrbracket$ is large enough to be likely to be strongly connected. Thus, any vertex in $\llbracket \neg \alpha \rrbracket$ has an infinite path avoiding $\alpha$; the only way for a vertex to satisfy $\mathsf{A}(\alpha'\mathsf{U}\alpha)$ is to satisfy $\alpha$ in the first place.
\end{proof}

This gives us the required ammunition to prove Theorem \ref{thm:ctl}.

\begin{proof}[Theorem \ref{thm:ctl}]
Given any formula $\psi$, we work bottom-up in its syntax tree, using Lemma \ref{lem:bustep} to get rid of temporal operators one by one. By Lemma \ref{lem:rep}, the propositional formula $\alpha$ we end up with has $\alpha \sim \psi$. Hence $\mathbb{K} \models \psi$ iff $I \subseteq \sem{\alpha}$. Since $|I| \rightarrow \infty$, and being initial is independent of satisfying $\alpha$, we have $\mathbb{P}(I \subseteq \sem{\alpha}) \rightarrow 0$ unless $\alpha$ is a tautology. For $\mathbb{K}^1$, we have $\ini \in \sem{\alpha}$ with probability $\varrho_{\alpha}$; this proves that this algorithm is correct.

To perform this algorithm, it is more effective to propagate $\mathcal{S}_{\alpha}$ (see (\ref{eq:Salpha})) rather than $\alpha$ itself; the former is equivalent to propagating the truth table of $\alpha$. The relevant operations we need to do on these truth tables at an operator are $\neg,\wedge,\vee$, and checking satisfiability/tautologies; all these take at most $2^{|A|}$ steps. Since the syntax tree has size $\mathcal{O}(|\psi|)$, we conclude that the total computation time is $\mathcal{O}(|\psi|)$.
\end{proof}

\begin{example}[Adapted from \cite{dong20250}]
Consider $\varphi = \mathsf{AG}(\mathsf{E}((\mathsf{AF} a_1)\mathsf{U}\neg a_1)) \wedge \mathsf{EG}(a_2)$. We repeatedly apply Lemma \ref{lem:bustep}, where the \textcolor{red}{red} part is the subformula being replaced.
\begin{align*}
 \mathsf{AG}(\mathsf{E}((\textcolor{red}{\mathsf{AF} a_1})\mathsf{U}\neg a_1)) \wedge \mathsf{EG}(a_2) &\sim \mathsf{AG}(\textcolor{red}{\mathsf{E}(a_1\mathsf{U}\neg a_1)}) \wedge \mathsf{EG}(a_2)  \\
&\sim \textcolor{red}{\mathsf{AG}(a_1 \vee \neg a_1)} \wedge \mathsf{EG}(a_2)  \\
&\sim \top \wedge \textcolor{red}{\mathsf{EG}(a_2)}  \\
&\sim \top \wedge a_2 \sim a_2.
\end{align*}
It follows that
\begin{align*}
\mathbb{P}(\mathbb{K} \models a_2) \rightarrow  0,& \textrm{ so } \mathbb{P}(\mathbb{K} \models \varphi) \rightarrow 0;\\
\mathbb{P}(\mathbb{K}^1 \models a_2) = p_2,& \textrm{ so } \mathbb{P}(\mathbb{K}^1 \models \varphi) \rightarrow p_2.
\end{align*}
\end{example}

\begin{remark}
Similar to LTL, our proof actually works for all choices of functions $p_{\recht{t}}(n)$ and $p_{\recht{i}}(n)$ for which $\liminf \tfrac{p_{\recht{t}}(n)}{\log n}$ is large enough and for which $p_{\recht{i}}(n) \gg n^{-1}$; hence our result generalizes that of \cite{dong20250}.
Unfortunately, our proof breaks down for small $c$: if $r+p_ac < 1$, then there is no propositional $\alpha$ such that $\mathsf{EX}a \sim \alpha$. Indeed, $|\sem{\mathsf{EX}a}| \approx n^{1-r-p_ac}$, and $\sem{\alpha}$ has linear size for all $\alpha \not \equiv \bot$. This shows that the approach of \cite{dong20250} and of this paper does not work, and that a radical new idea is needed to investigate 0-1 laws in this regime. Note also that contrary to LTL, for a CTL formula $\varphi$ the property $K \models \varphi$ is in general not monotonous: adding transitions makes it easier for $\mathsf{EX}a$ to be satisfied, but harder for $\mathsf{AX}b$, and so $K \models \mathsf{EX}a \wedge \mathsf{AX}b$ is not monotonous. Therefore, it is not clear whether CTL formulas have treshold functions.
\end{remark}

\section{Conclusion}

In this paper, we developed realistic probabilistic models for large Kripke structures, based on both theoretical knowledge of how large KSs originate, and empirical data from the Model Checking Contest \cite{amat2023behind}. These models form a stepping stone for the theoretical analysis of typical KS behaviour, and can be used to generate large, realistic benchmarks for evaluating model checkers. Furthermore, we study the asymptotic behaviour of these models under LTL and CTL, proving 0-1 laws and convergence laws, studying the theoretical complexity of determining limit probabilities, and giving algorithms for these problems. This gives a complete theoretical picture of the asymptotic behaviour of large KS under LTL and CTL.

A natural step for future research would be to study the behaviour of model checking algorithms on these models. Given a probabilistic model, one can consider the \emph{average-case complexity} of model checking algorithms. This has been studied for first-order logic model checking on random graphs \cite{dreier2020first}; it would be interesting to see to which extent these results carry over to our models.

Alternatively, our models can become more involved to model large KS even more realistically. Since large KS arise as products of small ones, it would be interesting to develop a model of KS based on Kronecker graphs \cite{leskovec2010kronecker}, a random graph model based on product graphs. Kronecker graphs behave, roughly, like Erd\H{o}s-Rényi graphs with inhomogeneity; hence this would also allow us to generalize away from the homogeneity assumption of Section \ref{sec:prelim}. Since our 0-1 laws and convergence laws do not depend on precise parameter choices, it is likely that our results would generalize to Kronecker-type models as well, though the proofs may be more tedious. 

Another important assumption to generalize away from is the Erd\H{o}s-Rényi assumption of independence, since this does not quite hold in product graphs: if one of the $s$ basic KS has an internal transition, this transition occurs $d^{s-1}$ in the product. In other words, a more realistic random model should have many `parallel' transitions. In this respect, exponential random graph models \cite{snijders2006new} are a promising tool.



\bibliography{sample}

\appendix
\appendix

\section{Proofs}

\subsection{Proof of Lemma \ref{lem:logdeg}}

By assumption 4, there is a $\sigma > 0$, independent of $s$, such that $\tfrac{c_m}{d^{s_m}} \geq \sigma$ for all $m \in M'$. It follows that
\begin{equation} \label{eq:logdeg1}
\sum_{m \in M}  \tfrac{c_m}{d^{s_m}}  \geq \sum_{m \in M'}  \tfrac{c_m}{d^{s_m}} \geq \sigma|M'| = \Theta(s).
\end{equation}
On the other hand, from assumption 2 we have 
\begin{equation*}
|M| \leq \sum_{i=1}^s M_i \leq Cs = \Theta(s);
\end{equation*}
hence, by assumption 3:
\begin{equation} \label{eq:logdeg2}
\sum_{m \in M}  \tfrac{c_m}{d^{s_m}} \leq D|M| = \mathcal{O}(s).
\end{equation}
From equations (\ref{eq:logdeg1}) and (\ref{eq:logdeg2}) we conclude $\sum_{m \in M}  \tfrac{c_m}{d^{s_m}} = \Theta(s)$, so
\begin{equation} \label{eq:logdeg3}
\mathsf{av}(\mathbb{K}) = \mathcal{O}(s).
\end{equation}
To get a lower bound for $\mathsf{av}(\mathbb{K})$, we need to know how tight the inequality $\mathsf{av}(\mathbb{K}) \leq \sum_{m \in M}  \tfrac{c_m}{d^{s_m}}$ is; hence we need to consider how many times a transition $t$ in $\mathbb{K}$ can be counted in $\sum_{m \in M} d^{s-s_m} c_m$, i.e., how many action labels $m$ can giv rise to $t$. If $t$ is not a self-loop in $\mathbb{K}$, then there is a $K_i$ in which $t$ is not a self-loop. Any action label $m$ that induces $t$ must satisfy $m \in M_i$. By assumption 2, there are at most $C$ such $m$. Furthermore, for such an $m$, we know that $t$ is not a self-loop in $\mathbb{K}_m$. It follows that the number of non-self loops in $\mathbb{K}$ is at least $C^{-1}\sum_{m \in M} d^{s-s_m}(c_m-f_m)$, hence
\[
\mathsf{av}(\mathbb{K}) \geq C^{-1}\sum_{m \in M} \frac{c_m-f_m}{d^{s_m}} \geq C^{-1}\varrho|M| = \Theta(s),
\]
where we use assumption 5. Together with \eqref{eq:logdeg3} this proves the Lemma.

\subsection{Proof of Theorem \ref{thm:KS}}

As remarked, we only need to prove that $\lim_{n \rightarrow \infty} (1-n^{-\alpha})^n = 0$. To show this, note that
\begin{align*}
\log((1-n^{-\alpha})^n) &= n \log(1-n^{-\alpha}) \\
&= n(-n^{-\alpha}+\mathcal{O}(n^{-2\alpha})) \\
&= -n^{1-\alpha}+\mathcal{O}(n^{1-2\alpha}) \\
&\rightarrow \infty,
\end{align*}
So $(1-n^{-\alpha})^n \rightarrow 0$. Note that the fact that $\alpha < 1$ is essential here.

\subsection{Proof of Theorem \ref{thm:01ltlmain}}

For $S \in 2^A$, let $\mathbb{K}^1_S$ be the random KS with the probability distribution of $\mathbb{K}^1$, conditioned on the fact that $\ell(\ini) = S$. Then $\mathbb{K}^1$ comes from a 2-stage random process, where we first select $S$ with probability $p_S$, and then generate $\mathbb{K}^1_S$. Clearly, if $\varphi$ is an $S$-tautology then $\mathbb{P}(\mathbb{K}^1_S \models \varphi) = 1$. We will show the following result:

\begin{proposition} \label{prop:K1S}
If $\varphi$ is not an $S$-tautology, then $\mathbb{P}(\mathbb{K}^1_S \models \varphi) \rightarrow 0$.
\end{proposition}

Indeed, this proves Theorem \ref{thm:01ltlmain}, since the case that $\varphi$ is an $S$-tautology is trivial.
To state the proof of Proposition \ref{prop:K1S}, we first introduce some notation. Throughout this appendix, we fix $S \in 2^A$.

\begin{definition}
Let $w = w_1\cdots w_k \in (2^A)^*$ be a (possibly empty) word.
\begin{enumerate}
\item If $w$ is nonempty, we define a \emph{$w$-road} in $\mathbb{K}$ to be sequence $(x_1,\ldots,x_k)$ of vertices in $\mathbb{K}^1_S$ such that each transition $x_i \rightarrow x_{i+1}$ exists, and such that $\ell(x_i) = w_i$ for all $i$.
\item If $w$ is nonempty and $x,y$ are vertices, then a \emph{$w$-path from $x$ to $y$} is a $w$-road $(x_1,\ldots,x_k)$ in $\mathbb{K}^1_S$ such that the transitions $x \rightarrow x_1$ and $x_k \rightarrow y$ exists.
\item If $w$ is empty, then a \emph{$w$-path} from $x$ to $y$ is a transition $x \rightarrow y$.
\item For a set $T \subseteq \mathbb{K}^1_S$ of vertices, we define $T_{w}$ to be the random graph with state set $T$, where a transition $x \rightarrow y$ exists in $T_{w}$ if and only if a $w$-path from $x$ to $y$ exists in $\mathbb{K}$.
\end{enumerate}
\end{definition}

In particular, for the empty word $\varepsilon \in (2^A)^*$ the graph $T_{\varepsilon}$ is just the subgraph of $\mathbb{K}$ induced by $T$. We make the following claim:

\begin{proposition} \label{prop:strongcon}
Let $S' \in 2^A$, and let $w^1,w^2 \in (2^A)^*$ such that $w^2$ is a nonempty word. Let $T = \{x \in \mathbb{K}^1_S \mid \ell(x) = S'\}$. Then a.a.s. there exists a $T' \subseteq T$ such that $T'_{w^1}$ is strongly connected, and there is a $w^2$-path from $\ini$ to some state in $T'$.
\end{proposition}

\begin{proof}[of Proposition \ref{prop:K1S} from Proposition \ref{prop:strongcon}] By Corollary \ref{cor:counterword}, there exist $u,v \in (2^A)^*$ such that $u$ and $v$ are both nonempty and $u_1 = S$. In fact, we may (and do) assume that $|u| > 1$, by replacing $u$ by $uv$ if necessary. Let $\breve{u},\breve{v} \in (2^A)^*$ be defined by $u = u_0\breve{u}$ and $v = v_1\breve{v}$. Take $S' = v_0$; then a.a.s. a suitable $T'$ exists, for $w^1 = \breve{v}$ and $w^2 = \breve{u}$ (note $|\breve{u}| >0$). Given such a $T'$, we can construct an infinite path in $\mathbb{K}^1_S$ from $h$, labeled $uv^{\omega}$, as follows: we start in $\ini$ (label $u_0$), take a $\breve{u}$-path into $T'$ (labels $\breve{u}v_0$). Because $T'_{\breve{v}}$ is strongly connected, we can indefinitely take $\breve{v}$-paths within $T'$ (labels $\breve{v}v_0$). This path corresponds to the infinite word $u_0\breve{u}v_0(\breve{v}v_0)^{\omega} = uv^{\omega}$. Hence a.a.s. a counterexample to $\varphi$ exists within $\mathbb{K}^1_S$, so $\mathbb{P}(\mathbb{K}^1_S \models \varphi) \rightarrow 0$.
\end{proof}

Our goal is now to prove Proposition \ref{prop:strongcon}. The idea behind the proof is to show that in $T_{w_1}$, each edge exists with probability $\Theta((\log n)^{|w_1|+1|}$. Furthermore, for large $n$, the existence of these edges can be approximated asymptotically as independent events, so $T_{w_1}$ behaves like an Erd\H{o}s-Rényi graph with $p = \Theta((\log n)^{|w^1|+1}/n)$. This is above the giant component threshold of $1/n$, so it has a giant component $T'_{w^1}$. This turns out to be large enough so that we can a.a.s. find a $w^2$-path from the initial vertex into it.

To make the intuition behind being `almost Erd\H{o}s-Rényi' more precise, we need a bit of technical work, starting with a few auxiliary lemmas.

\begin{lemma} \label{lem:noedges}
Let $\alpha,\beta > 0$. Then
\[
\mathbb{P}\left(\exists T,U \subseteq \mathbb{K}^1_S \textrm{ s.t. } \substack{ |T| = \lceil \alpha n\rceil, |U| = \lceil \beta n \rceil,\\\textrm{ and there are no transitions from $T$ to $U$}}\right) \rightarrow 0.
\]
\end{lemma}

\begin{proof}
For two sets $T,U$ of the prescribed size, let $C_{T,U}$ be the indicator of the event that there are no transitions from $T$ to $U$, and let $X = \sum_{T,U} C_{T,U}$. Note that $T$ and $U$ are allowed to overlap. We have 
\begin{align*}
\mathbb{P}(C_{T,U} = 1) &= (1-p_{\recht{t}})^{\lceil \alpha n\rceil\lceil \beta n\rceil}\\
&= \operatorname{exp}\left[ \lceil\alpha n \rceil \lceil \beta n \rceil\log(1-p_{\recht{t}})\right]\\
&= \operatorname{exp}\left[-\alpha\beta n^2p_{\recht{t}} - \tfrac{\alpha\beta}{2}n^2p_{\recht{t}}^2  + \mathcal{O}(p_{\recht{t}}+n^2p_{\recht{t}}^3) \right]\\
&= \operatorname{exp}\left[-\alpha\beta c n \log n - \frac{\alpha \beta c^2}{2}\log(n)^2 + \mathcal{O}\left(\tfrac{\log(n)}{n}+\tfrac{\log(n)^3}{n}\right) \right]\\
&= \operatorname{exp}\left[-\alpha\beta c n \log n - \frac{\alpha \beta c^2}{2}\log(n)^2\right](1+o(1)).
\end{align*}
The number of choices for $(U,T)$ is equal to $\binom{n}{\lceil \alpha n\rceil}\binom{n}{\lceil \beta n\rceil}$, which we can bound via the entropy bound of binomial coefficients:
\begin{align*}
\binom{n}{\lceil \alpha n\rceil}\binom{n}{\lceil \beta n\rceil} &\leq \operatorname{exp}\left[nh\left(\tfrac{\lceil \alpha n\rceil}{n}\right)+nh\left(\tfrac{\lceil \beta n\rceil}{n}\right)\right]\\
&= \operatorname{exp}\left[n(h(\alpha)+h(\beta))\right](1+o(1)),
\end{align*}
where $h$ is the binary entropy function. It follows that 
\[
\mathbb{E}[X] = \operatorname{exp}\left[n(h(\alpha)+h(\beta))-\alpha\beta c n \log n - \frac{\alpha \beta c^2}{2}\log(n)^2\right].
\]
Since the leading term here is the $n \log n$, which has a negative coefficient, we conclude that $\mathbb{E}[X] \rightarrow 0$, hence $\mathbb{P}(X > 0) \rightarrow 0$. This proves the lemma.
\end{proof}

\begin{lemma} \label{lem:matching}
Let $S,T$ be (possibly random) sets of vertices of $\mathbb{K}^1_S$, such that a.a.s. $|T| = (\alpha + o(1))n$, $|U| = (\beta + o(1))n$ for some $\alpha,\beta > 0$. Then a.a.s. there are $(\min(\alpha,\beta)+o(1))n$ edges from $T$ to $U$ such that all sources are different, and all targets are different.
\end{lemma}

\begin{proof}
Suppose that $\alpha < \beta$, so a.a.s. $|T| < |U|$ (the cases $\alpha = \beta$ and $\alpha > \beta$ are analogous). By Hall's marriage theorem, the maximum number of edges from $T$ to $U$ we can find without replicating sources or targets is
\[
|T|-\max_{T' \subseteq T}(|T'|-|N_{\text{out}}(T')|),
\]
where $N_{\text{out}}(T')$ is the number of out-neighbours of $T'$ in $U$. Let $\varepsilon > 0$. By Lemma \ref{lem:noedges}, a.a.s. each $T'$ with $|T'| \geq \varepsilon n$ has $|N_{\text{out}}(T')| = (1+o(1))|U| \geq |T'|$; hence these $T'$ do not contribute to $\max_{T' \subseteq T}(|T'|-|N_{\text{out}}(T')|)$. Since this holds for all $\varepsilon > 0$, it follows that a.a.s. $|T|-\max_{T' \subseteq T}(|T'|-|N_{\text{out}}(T')|) = (\alpha+o(1))n$, as was to be shown.
\end{proof}

\begin{definition}
Let $w \in (2^A)^*$ be nonempty. For a $w$-road $(x_1,\ldots,x_k)$ we define its \emph{start} to be $x_1$, and its \emph{end} to be $x_k$.
\end{definition}

\begin{lemma} \label{lem:diffroads}
Let $w = w_1\cdots w_k \in (2^A)^*$ be nonempty, and let $\alpha = \min_{i \leq k} \varrho_{w_i}$. Then a.a.s. there are $(\alpha+o(1))n$ $w$-roads in $\mathbb{K}^1_S$, such that the starts are all different, and the ends are all different.    
\end{lemma}

\begin{proof}
We prove this by induction on $k$. If $k = 1$, this just amounts to saying that there are $(\varrho_1+o(1))n$ vertices $x$ in $\mathbb{K}^1_S$ with $\ell(x) = w_1$, which is true a.a.s. Now suppose that $k > 1$, and the statements holds for $k-1$. Then by the induction hypothesis, a.a.s. there are $(\min_{1<i\leq k}\varrho_{w_i}+o(1))n$ ($w_2\cdots w_k$)-roads in $\mathbb{K}^1_S$ such that their starts are all different, and their ends are all different. Let $T$ be the set of these starts, and let $U$ be the set of $w_1$-labeled vertices in $\mathbb{K}^1_S$. Since a.a.s. $|T| = (\varrho_{w_1}+o(1))n$, we can use Lemma \ref{lem:matching} to a.a.s. find $(\min_{1\leq i\leq k}\varrho_{w_i})+o(1))n$ pairs $(x,y)$ of vertices in $T \times U$, such that all edges $x \rightarrow y$ exist, and such that the $x$ are all different and the $y$ are all different. If we compose these edges with the corresponding $(w_2\cdots w_k)$-road for each $y \in U$, we find the required $w$-roads.
\end{proof}
We are now in a position to prove the first part of Proposition \ref{prop:strongcon}: that there is a large subset  of $T$ that is strongly connected via $w^1$-paths.

\begin{lemma} \label{lem:strongcon}
Let $S' \in 2^A$, and let $w \in (2^A)^*$. Let $T = \{x \in \mathbb{K}^1_S \mid \ell(x) = S'\}$. Then a.a.s. there exists a $T' \subseteq T$ such that $|T'| = (\varrho_{S'}+o(1))n$ and such that $T'_{w}$ is strongly connected.
\end{lemma}

\begin{proof}
If $w$ is the empty word, then $T_w$ is the induced subgraph of $T$ in $\mathbb{K}^1_S$. Let $0 < \varepsilon < \tfrac{1}{2}$; we will show that 
\[
\mathbb{P}(T \textrm{ has a strongly connected component of size $>(1-\varepsilon)\varrho_{S'} n$}) \rightarrow 1.
\]
Indeed, suppose that this is not the case, i.e., the largest strongly connected component has size at most $(1-\varepsilon)\varrho_{S'} n$. Let $D$ be the directed acyclic graph of strongly connected components in $T$. Construct a subset $D' \subseteq D$ recursively as follows: we start with $D' = \varnothing$, and we keep adding minimal elements of $D \setminus D'$ to $D'$. We stop when the SCCs in $D'$ have total size $\geq \tfrac{\varepsilon\varrho_{S'}}{2}n$. Since the largest SCC has size at most $(1-\varepsilon)\varrho_{S'} n$, the total size of the SCCs in $D'$ is $< (1-\tfrac{\varepsilon}{2})\varrho_{S'} n$. Let $S$ be the union of these SCCs, and let $S'$ be its complement. Then by construction there are no edges from $S$ to $S'$. Since a.a.s. $|S'| \geq (\tfrac{\varepsilon\varrho_{S'}}{2}+o(1))n$ and $|S| \geq \tfrac{\varepsilon \varrho_{S'}}{2}n$, this a.a.s. never happens by Lemma \ref{lem:noedges}. We conclude that our assumption is a.a.s. never true, i.e. $\mathbb{P}(T \textrm{ has a strongly connected component of size $> (1-\varepsilon)\varrho_{S'} n$}) \rightarrow 1$, as we aimed to show.

Now suppose that $w = w_1\cdots w_k$ for some $k \geq 1$. There is a $\beta > 0$ such that a.a.s. there exists a set $\mathcal{I}$ of $w$-roads in $\mathbb{K}^1_S$ of size $\geq \beta n$ , such that all starts are different, and all ends are different. Let $\{s_i\}_{i \in \mathcal{I}}$ be the set of starts, and let $\{e_i\}_{i \in \mathcal{I}}$ be the set of ends. For $x,y \in T$, we let $E_{xy}$ be the indicator of the event that the transition $x \rightarrow y$ exists in $\mathbb{K}^1_S$, and we let $E'_{xy}$ to be a new Bernoulli random variable with probability $c\log n/n$, such that all $E'_{xy}$ are independent from each other and from all random variables so far. Define a new random variable $\hat{E}_{xy}$ by
\[
\hat{E}_{xy} = \begin{cases}E_{xy}, & \textrm{ if $x \rightarrow y$ is not an edge in any $w$-road in $\mathcal{I}$},\\
E'_{xy}, & \textrm{ otherwise.}
\end{cases}
\]
Note that this means that $\hat{E}_{xy} \leq E_{xy}$. Now consider the graph $\hat{T}_w$ as follows: its vertices are the vertices of $T$, and there is an edge $x \rightarrow y$ in $T_l$ if and only if there is a $w$-road $i$ in $\mathcal{I}$, and $\hat{E}_{xs_i} = \hat{E}_{t_iy} = 1$. Note that $\hat{T}_l$ is a subgraph of $T_l$, and that $\hat{T}_l$ is an Erd\H{o}s-Rényi graph with a.a.s. $(\varrho_{S'}+o(1))n$ vertices. The probability of an edge $x \rightarrow y$ is a.a.s.
\begin{align*}
\mathbb{P}(x \rightarrow y \textrm{ in $\hat{T}_l$}) &\geq 1-(1-p_{\recht{t}}^2)^{\beta n}\\
&= 1-\operatorname{exp}\left[\beta n\log(1-p_{\recht{t}}^2)\right]\\
&= 1-\operatorname{exp}\left[-\beta np_{\recht{t}}^2 + \mathcal{O}(np_{\recht{t}}^4)\right]\\
&= 1-\operatorname{exp}\left[\tfrac{-\beta c (\log n)^2}{n}+\mathcal{O}\left(\tfrac{(\log n)^4}{n^3}\right)\right]\\
&= 1- \left(1-\tfrac{\beta c (\log n)^2}{n}+\mathcal{O}\left(\tfrac{(\log n)^4}{n^3}\right)\right)+\mathcal{O}\left(\tfrac{(\log n)^4}{n^2}\right)\\
&= \tfrac{\beta c (\log n)^2}{n} +\mathcal{O}\left(\tfrac{(\log n)^4}{n^2}\right).
\end{align*}
This is above the threshold value of $\frac{\log n}{n}$ of strong connectivity for Erd\H{o}s-R\'enyi graphs, hence $\hat{T}_l$ is strongly connected a.a.s. Since $T_l$ only has more edges, this means that a.a.s. $T_l$ is strongly connected as well, so we can take $T'= T$.
\end{proof}

To finish the proof of Proposition \ref{prop:strongcon}, we just need to show that the initial state a.a.s. has a $w^2$-road into this $T'$:

\begin{lemma}
Let $S' \in 2^A$, and let $w \in (2^A)^*$. Let $T = \{x \in \mathbb{K}^1_S \mid \ell(x) = S'\}$, and let $T' \subseteq T$ be such that a.a.s. $|T'| = (\varrho_{S'}+o(1))n$. Then a.a.s. there is a $w$-path from the initial state of $\mathbb{K}^1_S$ into $T'$.
\end{lemma}

\begin{proof}
In the proof of Lemma \ref{lem:strongcon}, the construction of $T'$ is independent of the set of initial states $I$, since the random set $I$ is defined independently from transitions and state labels. Therefore, we can obtain $\mathbb{K}^1_S$ by first generating $\mathbb{K}$, and then uniformly selecting a state in $\{x \in \mathbb{K} \mid \ell(x) = S\}$ to be the initial state. 

Now let $\varepsilon > 0$; we will prove that 

\begin{equation} \label{eq:liminf}
\liminf \mathbb{P}(\exists \textrm{ $w$-path from $\ini$ into $T'$}) \geq 1-\varepsilon;
\end{equation}
this proves that a.a.s. such a path exists. By Lemma \ref{lem:diffroads} there exists an $\alpha > 0$ such that a.a.s. there are $\geq \alpha n$ different $wS'$-roads in $\mathbb{K}$ such that the starts and ends are all different. These ends all lie in $T$; by discarding $o(n)$ roads ending in $T \setminus T'$ and amending $\alpha$ if necessary, we a.a.s. find a set $\mathcal{I}$ of $\geq \alpha n$ different $wS'$-roads whose starts are all different and whose ends are all different and lie in $T'$. Now consider the set $\mathcal{X}$ of the $\geq \alpha n$ starts of these roads. By Lemma \ref{lem:noedges}, a.a.s. there is no set of size $\geq \frac{\varepsilon \varrho_S}{2} n$ without transitions into $\mathcal{X}$. In particular, there are a.a.s. $> (1-\varepsilon)n\varrho_S$ states with label $S$ and with a transition into $\mathcal{X}$; these states have a $w$-path into $T'$. Since the initial state is drawn uniformly from the $\approx n\varrho_{S}$ states with label $S$, this proves \eqref{eq:liminf}.
\end{proof}

\subsection{Proof of Theorem \ref{thm:01ltl}}

As usual, we only care about the case that $\varphi$ is not a tautology; then there is an $S \in 2^A$ such that $\varphi$ is not an $S$-tautology. With high probability, there exists an $x \in I$ such that $\ell(x) = S$; suppose that this is the case. Let $\hat{\mathbb{K}}$ be the KS obtained from $\mathbb{K}$ with set of initial states $\hat{I} = \{x\}$. Then $\hat{\mathbb{K}}$ is an instance of $\mathbb{K}^1_S$, hence by Proposition \ref{prop:K1S} we have $\mathbb{P}(\hat{\mathbb{K}} \models \varphi) \rightarrow 0$. On the other hand, since $\hat{\mathbb{K}}$ and $\mathbb{K}$ are identical except for their initial sets and $\hat{I} \subseteq I$, we have $\mathbb{P}(\hat{\mathbb{K}} \models \varphi) \geq \mathbb{P}(\mathbb{K} \models \varphi)$. Together, these two results prove the Theorem.

\subsection{Example \ref{ex:ctl} explained} \label{app:ex}

First, note that
\begin{align*}
\mathbb{E}[X] &= np_{\recht{i}}(1-p_ap_{\recht{t}})^n\\
&= n^{1-r}\left(1-p_acn^{-1}\log(n)\right)^n \\
&= \operatorname{exp}\left[(1-r)\log(n) + n\log(1-p_acn^{-1}\log(n))\right]\\
&= \exp\left[(1-r)\log(n)-n\left(-p_acn^{-1}\log(n) +\mathcal{O}(n^{-2}\log(n)^2)\right)\right]\\
&= \exp\left[(1-r-p_ac)\log(n)+\mathcal{O}(n^{-1}\log(n)^2)\right]\\
&\rightarrow \begin{cases}
    0, & \textrm{ if } r+p_ac > 1,\\
    1, & \textrm{ if } r+p_ac = 1,\\
    \infty, & \textrm{ if } r+p_ac< 1.
\end{cases}
\end{align*}

The case $r+p_ac > 1$ was already handled in Example \ref{ex:ctl}. Now suppose $r+p_ac < 1$, and let $\varepsilon = 1-r-p_ac$. Then $\mathbb{E}[X] = (1+o(1))n^{\varepsilon}$. Furthermore, the $C_x$ are independent Bernoulli variables with probability $\varrho := p_{\recht{i}}(1-p_ap_{\recht{t}})^n = (1+o(1))n^{\varepsilon-1}$, so
\begin{align*}
\operatorname{Var}(X) &= n\varrho(1-\varrho)\\
&= np_{\recht{i}}(1-p_ap_{\recht{t}})^n(1-p_{\recht{i}}(1-p_ap_{\recht{t}})^n)\\
&= (1+o(1))n^{\varepsilon}(1-n^{\varepsilon-1})\\
&= o(n^{2\varepsilon}) = o(\mathbb{E}[X]^2).
\end{align*}
By the second moment method, this means that $\mathbb{P}(X = 0) \rightarrow 0$.

Now suppose $r+p_ac = 1$. We will show that $X$ converges in distribution to a Poisson random variable by showing $\mathbb{E}[(X)_k] \rightarrow 1$ for all $k \in \mathbb{Z}_{\geq 1}$, where $(X)_k = X(X-1)\cdots(X-k+1)$. Note that $\mathbb{E}[(X)_k]$ is the expected number of $k$-tuples of different vertices $(x_1,\ldots,x_k)$ such that $C_{x_1}=\cdots = C_{x_k} = 1$. Again, the $C_x$ are independent Bernoulli random variables, with probability $\varrho = p_{\recht{i}}(1-p_ap_{\recht{t}})^n = (1+o(1))n^{-1}$. For each tuple we have
\begin{align*}
\mathbb{P}(C_{x_1} = \cdots = C_{x_k} = 1) &= \varrho^k \\
&= (1+o(1))n^{-k}.
\end{align*}
Since there are $\Theta(n^k)$ of such tuples, we get $\mathbb{E}[(X)_k] = 1+o(1)$, as was to be shown. Hence $X$ converges to a Poisson distribution with parameter $1$, and $\mathbb{P}(X=0) \rightarrow \textrm{e}^{-1}$.

\subsection{Proof of Lemma \ref{lem:bustep}}

Clearly, $\sem{\mathsf{EX}\alpha} \subseteq \sem{\mathsf{EF}\alpha}$, so it suffices to consider $\mathsf{EF}\alpha$ for $\alpha \equiv \bot$ and $\mathsf{EX}\alpha$ $\alpha \not \equiv \bot$.
If $\alpha \equiv \bot$, then clearly $\sem{\mathsf{EF}\alpha} = \varnothing = \sem{\bot}$. On the other hand, suppose that $\alpha \not \equiv \bot$, i.e, $\varrho_{\alpha} > 0$. The assumption on $c$ ensures that $\varrho_{\alpha}c > 1$, so certainly $r+\varrho_{\alpha}c > 1$. The same argument as in Example \ref{ex:ctl} shows that $\sem{\mathsf{EX}\alpha} = V = \sem{\top}$ a.a.s.

The result for $\mathsf{AX}$ and $\mathsf{AG}$ follows directly, since $\mathsf{AX}\alpha \equiv \neg \mathsf{EX}\neg \alpha$ and $\mathsf{AG}\alpha \equiv \neg \mathsf{EF}\neg \alpha$.

We now consider $\mathsf{A}(\alpha'\mathsf{U}\alpha)$; this also proves the result for $\mathsf{AF}$ and $\mathsf{EG}$ since $\mathsf{AF}\alpha \equiv \mathsf{A}(\top\mathsf{U}\alpha)$ and $\mathsf{EG}\alpha \equiv \neg \mathsf{AF}\neg\alpha$. If $\alpha \equiv \top$, then clearly $\sem{\mathsf{A}(\alpha'\mathsf{U}\alpha)} = V = \sem{\top}$; and if $\alpha \equiv \bot$, then $\sem{\mathsf{A}(\alpha'\mathsf{U}\alpha)} = V = \sem{\bot}$. Now suppose $\alpha \not \equiv \top,\bot$, so $0 < \varrho_{\alpha} < 1$; we claim that a.a.s. $\sem{\alpha}$ and $\sem{\neg\alpha}$ are strongly connected. To see this for $\sem{\alpha}$, let $0 < \varepsilon < 1-(r+\varrho_{\alpha}c)^{-1}$; then
\[
\mathbb{P}\Big[(1-\varepsilon)n\varrho_{\alpha} < |\sem{\alpha}| < (1+\varepsilon)n\varrho_{\alpha}\Big] \rightarrow 1.
\]
Suppose that $|\sem{\alpha}|$ is within these bounds, so $\frac{|\sem{\alpha}|}{(1+\varepsilon)\varrho_{\alpha}} < n < \frac{|\sem{\alpha}|}{(1-\varepsilon)\varrho_{\alpha}}$; then the subgraph of $\mathbb{K}$ induced by $\sem{\alpha}$ is a random directed graph, with $\rightarrow \infty$ vertices, and edge probability
\begin{align*}
\frac{c\log n}{n} &\geq \frac{c\log\left(\frac{|\sem{\alpha}|}{\varrho_{\alpha}(1+\varepsilon)}\right)}{\frac{|\sem{\alpha}|}{(1-\varepsilon)\varrho_{\alpha}}}\\
&= \frac{(1-\varepsilon)\varrho_{\alpha}c\log(|\sem{\alpha}|)-c\log(\varrho_{\alpha}(1+\varepsilon))}{|\sem{\alpha}|}.
\end{align*}
By the assumption on $\varepsilon$ we have $(1-\varepsilon)\varrho_{\alpha}c > 1$. As a result, the edge probability is larger than the threshold for strong connectivity \cite{graham2008note}, and $\sem{\alpha}$ is strongly connected a.a.s.; the proof for $\sem{\neg \alpha}$ is analogous. 

Now let $x$ be any state. If $x \in \sem{\alpha}$, then clearly $x \in \sem{\mathsf{A}(\alpha'\mathsf{U}\alpha)}$. However, if $x \in \sem{\neg \alpha}$, then a.a.s. there is an infinite path from $x$ entirely within $\sem{\neg \alpha}$, as this is strongly connected. Hence $x \notin \sem{\mathsf{A}(\alpha'\mathsf{U}\alpha)}$. We conclude that $\mathsf{A}(\alpha'\mathsf{U}\alpha) \sim \alpha$.

Finally, consider $\mathsf{E}(\alpha'\mathsf{U}\alpha)$; the case that $\alpha \equiv \bot$ is trivial. Note that $\alpha \vee (\alpha' \wedge \mathsf{EX}\alpha)$ implies $\mathsf{E}(\alpha'\mathsf{U}\alpha)$, and $\mathsf{E}(\alpha'\mathsf{U}\alpha)$ implies $\alpha \vee \alpha'$. Hence
\[
\sem{\alpha \vee (\alpha' \wedge \mathsf{EX}\alpha)} \subseteq \sem{\mathsf{E}(\alpha'\mathsf{U}\alpha)} \subseteq \sem{\alpha \vee \alpha'}.
\]
Since $\mathsf{EX}\alpha \sim \top$, we have $\alpha \vee (\alpha' \wedge \mathsf{EX}\alpha) \sim \alpha \vee \alpha'$. This shows that a.a.s. $\sem{\mathsf{E}(\alpha'\mathsf{U}\alpha)} = \sem{\alpha \vee \alpha'}$, as was to be shown.

\end{document}